\documentclass[preprint,11pt]{elsarticle}

\usepackage{amsmath,amssymb}
\usepackage{booktabs}
\usepackage{graphicx}
\usepackage{siunitx}
\usepackage{amsthm}
\usepackage{algorithm}
\usepackage{algpseudocode}
\usepackage[hidelinks]{hyperref}
\usepackage{xurl}
\usepackage{microtype}

\newtheorem{lemma}{Lemma}
\newtheorem{theorem}{Theorem}
\newtheorem{proposition}{Proposition}
\theoremstyle{remark}
\newtheorem{remark}{Remark}

\newcommand{\rmin}{r_{\min}}
\newcommand{\rmax}{r_{\max}}
\newcommand{\Tmat}{\mathbf{T}}
\newcommand{\Umat}{\mathbf{U}}
\newcommand{\Order}{\mathcal{O}}
\DeclareMathOperator{\acos}{arccos}

\newcommand{\confAspectOblateSpheroid}{3.8}
\newcommand{\confLocalOrderBelowRc}{4}
\newcommand{\confLocalOrderAboveRc}{2.3}
\newcommand{\confLocalWindowAboveRc}{0.07}
\newcommand{\confOrderFourth}{4.0}

\newcommand{\confOrderPrismProduct}{1.6}
\newcommand{\confOrderPrismMobius}{2.5}
\newcommand{\confOrderSpheroidNosub}{1.5}
\newcommand{\confOrderSpheroidSubTwo}{3.96}
\newcommand{\confOrderCylNosubTwo}{1.50}
\newcommand{\confOrderCylNosubOne}{1.5}
\newcommand{\confZenithPlainSixForty}{2e-5}
\newcommand{\confMobiusProductAgree}{1e-6}
\newcommand{\confOrderBulletNosub}{2.8}

\newcommand{\confCscaStar}{0.5413924}
\newcommand{\confKrmaxFullLadderSlope}{3.5}

\newcommand{\confKrmaxOrderOnePtTwoFive}{4.00}
\newcommand{\confKrmaxOrderTwoPtFive}{4.09}
\newcommand{\confKrmaxOrderFive}{3.99}
\newcommand{\confKrmaxOrderTen}{4.00}
\newcommand{\confAspectTallCylinder}{4.1}
\newcommand{\confAspectFlatPrism}{3.8}
\newcommand{\confAspectTallPrism}{3.7}
\newcommand{\confArcEndpointMatch}{7\times10^{-6}}

\newcommand{\confOnsetTwoPtFive}{12}
\newcommand{\confOnsetTen}{16}
\newcommand{\confOnsetTwenty}{24}

\journal{Journal of Quantitative Spectroscopy and Radiative Transfer}
\makeatletter
\def\ps@pprintTitle{%
  \let\@oddhead\@empty \let\@evenhead\@empty
  \let\@oddfoot\@empty \let\@evenfoot\@empty}
\makeatother

\begin{document}

\begin{frontmatter}

\title{Boundary-conformal integration for the invariant-imbedding
T-matrix method: high-order convergence for faceted particles}

\author[pku]{Zihua Wu\corref{cor}}
\ead{wuzihua@pku.edu.cn}
\author[seu]{Yu Xiong}
\cortext[cor]{Corresponding author.}
\address[pku]{School of Earth and Space Sciences, Peking University,
Beijing 100871, China}
\address[seu]{State Key Laboratory of Millimeter Waves, School of Information Science
and Engineering, Southeast University, Nanjing 210096, China}

\begin{abstract}
The invariant-imbedding T-matrix method (IITM) is a standard tool for light scattering by
large, sharply faceted, non-axisymmetric particles (atmospheric ice crystals and mineral
dust) where the surface-based extended boundary condition method loses accuracy. Its accuracy
is limited by ``staircasing'': the dielectric contrast of a faceted particle is integrated
across boundaries that cut the quadrature grid, so standard quadrature converges at low
algebraic order. We show that this non-smoothness has a single geometric origin, the tangencies
of the integration sphere to the faces and edges of the particle, which produce
\emph{jumps, kinks, and half-integer branches} according to the tangency type, in all three
integration directions. A boundary-conformal scheme removes them using closed-form azimuthal
coefficients, panel splitting at the analytically known tangency loci, and a square-root
substitution $x\mapsto x_c+t^2$ that absorbs the half-integer branches. For a hexagonal prism
the azimuthal integration becomes exact and the zenithal and radial directions recover spectral
and fourth-order convergence; because the construction depends only on the contact geometry, it
extends to any convex polyhedron, demonstrated on the solid hexagonal bullet (a faceted ice
habit with tilted faces). The zenithal crossing is a square-root branch rather than a kink, so
the established interval-splitting alone gives only $\Order(N^{-3})$, while the radial step
removes the half-integer edge branch that caps the Riccati recurrence on faceted particles. The
convergence orders are fixed by the local contact geometry and verified size-independent up to
$k\rmax=20$; what grows with size is the resolution needed to reach each asymptotic regime, not
the order.
\end{abstract}


\end{frontmatter}

\section{Introduction}
\label{sec:intro}

The T-matrix linearly relates the expansion coefficients of the scattered field to
those of the incident field and encodes the full single-scattering response of a
particle \citep{waterman1971,mishchenko2002}. For non-spherical particles it is
classically computed by the extended boundary condition method (EBCM, or null-field method), which solves a single surface
integral equation and is spectrally accurate for smooth, moderate-aspect shapes but
becomes ill-conditioned for large size parameters, extreme aspect ratios, and sharp
edges. Recent surface-integral-equation T-matrix formulations restore stability at large
aspect ratio for smooth penetrable particles~\citep{ganeshhawkins2025}, but do not address
the faceted case. The invariant-imbedding T-matrix method (IITM) \citep{johnson1988,bi2013,bi2014,sun2020iitm}
trades the single surface integral for a volume formulation: the particle is tiled
into concentric spherical shells and the T-matrix is grown outward from an inner
core by a matrix Riccati recursion in the shell radius $r$, with the angular
coupling on each shell obtained from a Fourier/quadrature analysis of the contrast.
IITM is robust where EBCM fails, and underpins community ice-crystal
single-scattering databases \citep{yang2013ice}. The present work develops the boundary-conformal high-order
quadrature within the open-source \texttt{TransitionMatrices.jl} framework
\citep{transitionmatrices}, extending its baseline IITM solver.

The price of the volume formulation is that the integrands are non-smooth wherever
the (spherical) integration grid crosses the (faceted) particle boundary. On a fixed
shell the contrast $\varepsilon(\vartheta,\varphi)-1$ is a piecewise-constant
function with jumps at the boundary crossings; along the radius the shell--particle
intersection changes character at the radii where the sphere becomes tangent to a
face or an edge. Standard equidistant (fast Fourier transform, FFT) or Gauss--Legendre quadrature of such
integrands converges only at low algebraic order, the ``staircasing'' that forces
IITM users toward very fine angular and radial grids. \citet{zhai2019} improved the
\emph{zenithal} integration for hexagonal prisms by \emph{splitting} the
$\vartheta$-integral at the shell--surface crossing $r\sin\vartheta=b$ (the
inscribed-cylinder radius) so that Gauss--Legendre is not applied across the discontinuity,
obtaining a substantial accuracy gain. They split at that crossing without characterizing its
branch type or deriving a convergence order.
\citet{doicu2019} analyzed
the \emph{radial} recurrence, recognizing it as a matrix Riccati equation and
integrating its linearized (Hamiltonian) form with a Pad\'e approximation of the
matrix exponential, demonstrated on smooth spheroids. Subsequent IITM work
has targeted efficiency and analysis rather than quadrature accuracy ---
$N$-fold-symmetry exploitation \citep{hu2021sym,zhao2022sym}, dimension-variable growth of the
recurrence \citep{hu2023dviim}, and analytic Jacobians and linearization
\citep{sun2021jac,sun2022liitm}; \citet{hu2020}
give an empirical Gauss-point-count rule for the zenith integral, and \citet{zhang2022vswf} a
phase-function criterion for the multipole truncation $n_{\max}$, but provide no
convergence-order analysis, branch-point identification, or smoothing substitution;
\citet{wang2023grid} add a flexible discrete-grid handling of arbitrary shape and internal
inhomogeneity. These works target the multipole truncation, symmetry reduction, linearization,
and shape/inhomogeneity flexibility; none addresses the geometric-branch structure of the
contrast \emph{quadrature} or the per-direction convergence orders treated here.
(For the broader T-matrix state of the art see the database update of
\citet{mishchenko2020db}.)

In this paper we argue that both improvements are facets of a single structure. The
non-smoothness in \emph{every} integration direction originates at the same
geometric events (the sphere--face and sphere--edge tangencies), and at each
such event the integrand acquires a \emph{jump, a kink, or a half-integer branch} in
the distance to the tangency, according to the tangency type. We use ``tangency'' loosely for the
family of contact events that seed the non-smooth branches; a single feature can act in more than
one direction: the top/bottom face, for instance, produces a \emph{zenithal jump} at
$\vartheta_{\mathrm{cap}}$ (the fixed-$r$ slice crossing the cap plane $z=\pm h/2$) and a
\emph{radial} $\sigma=1$ kink at the genuine sphere--face tangency $r=h/2$, both tracing to the
same face. Concretely:
\begin{itemize}
\item \textbf{Azimuthal.} On a shell that cuts the lateral faces, the contrast is a
  boxcar in $\varphi$; its Fourier coefficients are available in closed form as a sum
  over the boundary $\varphi$-crossings (no quadrature error).
\item \textbf{Zenithal.} As the polar angle $\vartheta$ sweeps, the boundary crossing
  $\vartheta_b(r)$ has a square-root branch where the shell is tangent to the
  inscribed cylinder; the integrand is analytic in $\sqrt{\vartheta-\vartheta_b}$, so
  $\vartheta = \vartheta_b + t^2$ removes it and Gauss--Legendre becomes spectral.
\item \textbf{Radial.} The shell--face tangency gives a slope kink (cured by panel
  splitting), while a vertical-\emph{edge} tangency gives a $(r-r_c)^{3/2}$ branch
  (analytic in $\sqrt{r-r_c}$); the same $t^2$ substitution restores fourth-order
  convergence of the Riccati integrator.
\end{itemize}
The unifying object is the square-root substitution $x\mapsto x_c+t^2$ that maps any
half-integer power $(x-x_c)^{k/2}$ to an analytic integer power in $t$. The substitution itself
is a standard device for removing algebraic endpoint and vertex singularities from Gauss
quadrature, long used in the boundary-element literature \citep{duffy1982,telles1987,davis1984};
our contribution is not the substitution but the \emph{a priori} identification --- from the
contact geometry alone --- of \emph{which} half-integer branch each tangency produces and
\emph{where}, so that the substitution is applied automatically and without tuning. We make this
precise (\S\ref{sec:branches}), give the integration scheme (\S\ref{sec:scheme}), and
demonstrate it on a hexagonal prism, a finite cylinder, a prolate spheroid, and, through
a convex-polyhedron form of the breakpoints (\S\ref{sec:polyhedra}), the solid hexagonal
bullet, a faceted ice habit with tilted faces (\S\ref{sec:results}).

\paragraph{Scope and prior work} The zenithal interval-splitting is due to
\citet{zhai2019}. Our additions are: (i) the crossing is a \emph{square-root branch}
(Lemma~\ref{lem:sqrt}), not the finite kink their model assumes, so splitting alone
gives only $\Order(N^{-3})$ (Proposition~\ref{prop:gauss}); (ii) the substitution
$\vartheta=\vartheta_b+t^2$ restores \emph{spectral} convergence
(Theorem~\ref{thm:zenith}); (iii) an explicit convergence-order analysis in all three
directions; and (iv) the azimuthal coefficients are taken in quadrature-free closed
form (Lemma~\ref{lem:arc}). At the accuracies reported by \citet{zhai2019} the angular
cost is in fact dominated by resolving the oscillation of high-order Wigner-$d$
functions ($N_q\approx600$--$1000$, the $q$-dependence of Remark~\ref{rem:unif}); the
square-root branch dominates only at higher accuracy, where the substitution is
decisive. The radial high-order recurrence
itself is \emph{not} new: \citet{doicu2019} already linearize the Riccati equation
and use a high-order matrix exponential. Our radial contribution is to show that on
\emph{faceted} particles (which \citet{doicu2019} do not treat; their tests are a
smooth spheroid and a sphere) this high-order recurrence is capped, without any visible
warning in the output, at
$\Order(N^{-5/2})$ by the half-integer edge branch, and that the $t^2$ substitution
restores its design order. The two together give a high-order IITM in all three
integration directions for the faceted particles that motivate the method.

\section{The IITM and its three quadratures}
\label{sec:background}

We summarize only what is needed; see \citet{johnson1988,doicu2019} for the full
derivation. The particle, of relative refractive index $m$, is enclosed in a ball of
radius $\rmax$; a homogeneous core of radius $\rmin$ --- the radius of the largest sphere
inscribed in the particle ($\rmin=\min(b,h/2)$ for the prism), distinct from the per-feature
critical radii of \S\ref{sec:radial} --- is initialized with the Mie T-matrix, and the
T-matrix is grown to $\rmax$. The
regular and outgoing radial blocks $\mathbf{J}(r)$, $\mathbf{H}(r)$ are built from
Riccati--Bessel functions, and the contrast coupling matrix is
\begin{equation}
  \label{eq:Udef}
  [\Umat(r)]_{pp'} \;=\; \int_{\mathbb S^2}\bigl[\varepsilon(r,\vartheta,\varphi)-1\bigr]\,
    \mathbf f_{p}(\vartheta,\varphi)\cdot\mathbf f_{p'}(\vartheta,\varphi)\,\mathrm d\Omega ,
\end{equation}
where the shell contrast $\varepsilon-1$ is projected onto products of the (normalized) vector
spherical wave functions (VSWFs) indexed by $p=(n,m,\text{parity})$. Both~\eqref{eq:Udef} and its zenithal
reduction~\eqref{eq:Uzenith} (introduced below) are written in schematic block form, suppressing the polarization indices, the
conjugation/parity bookkeeping, and overall constants; the full electromagnetic block algebra is
standard (see \citet{johnson1988,bi2013,doicu2019}). In particular the product
$\mathbf f_p\cdot\mathbf f_{p'}$ carries the appropriate complex conjugation, with the convention
fixed so that the surviving azimuthal order is $q=m'-m$. With this coupling, the imbedding recursion
is, in the continuous limit, the matrix Riccati equation
\begin{equation}
  \label{eq:riccati}
  \frac{d\Tmat}{dr}
  \;=\; \mathrm{i}k\,\bigl(\mathbf{J}^{\!\top} + \Tmat\,\mathbf{H}^{\!\top}\bigr)\,
        \Umat(r)\,\bigl(\mathbf{J} + \mathbf{H}\,\Tmat\bigr),
  \qquad \Tmat(\rmin)=\Tmat_{\mathrm{Mie}} .
\end{equation}
Each entry of $\Umat(r)$ separates into the two angular quadratures: the $\varphi$-integral of
\eqref{eq:Udef} selects the single azimuthal order $q=m'-m$, leaving a one-dimensional zenithal
integral
\begin{equation}
  \label{eq:Uzenith}
  [\Umat(r)]_{pp'} \;=\; \int_0^{\pi} g_{pp'}(\vartheta)\,c_{m'-m}(r,\vartheta)\,
    \sin\vartheta\,\mathrm d\vartheta ,
\end{equation}
where $c_q(r,\vartheta)$ is the one-period azimuthal Fourier coefficient of the contrast
(\S\ref{sec:azimuth}; closed form \eqref{eq:arc}) and the weight $g_{pp'}(\vartheta)$ is a product
of the angular functions $\pi_{mn},\tau_{mn},d^{\,n}_{0m}$ (the polarization and block structure
contributes a few such terms, all at the same $q=m'-m$). Computing a scattering observable therefore
requires three nested integrations: the \emph{azimuthal} ($\varphi\!\to c_q$) and
\emph{zenithal} ($\vartheta$) integrals that build $\Umat(r)$ on each shell, and the
\emph{radial} integration of \eqref{eq:riccati}. For an $N$-fold particle the azimuthal integral
is taken over one period $[0,2\pi/N]$ and only Fourier orders $q$ that are multiples of $N$
survive.
Figure~\ref{fig:schematic} shows how, as the shell radius $r$ grows, the integration
sphere becomes tangent to the prism at distinct feature types, each of which seeds the
non-smooth branches analyzed below in one or more of the three quadratures.

\begin{figure}[htbp]
  \centering
  \includegraphics[width=\textwidth]{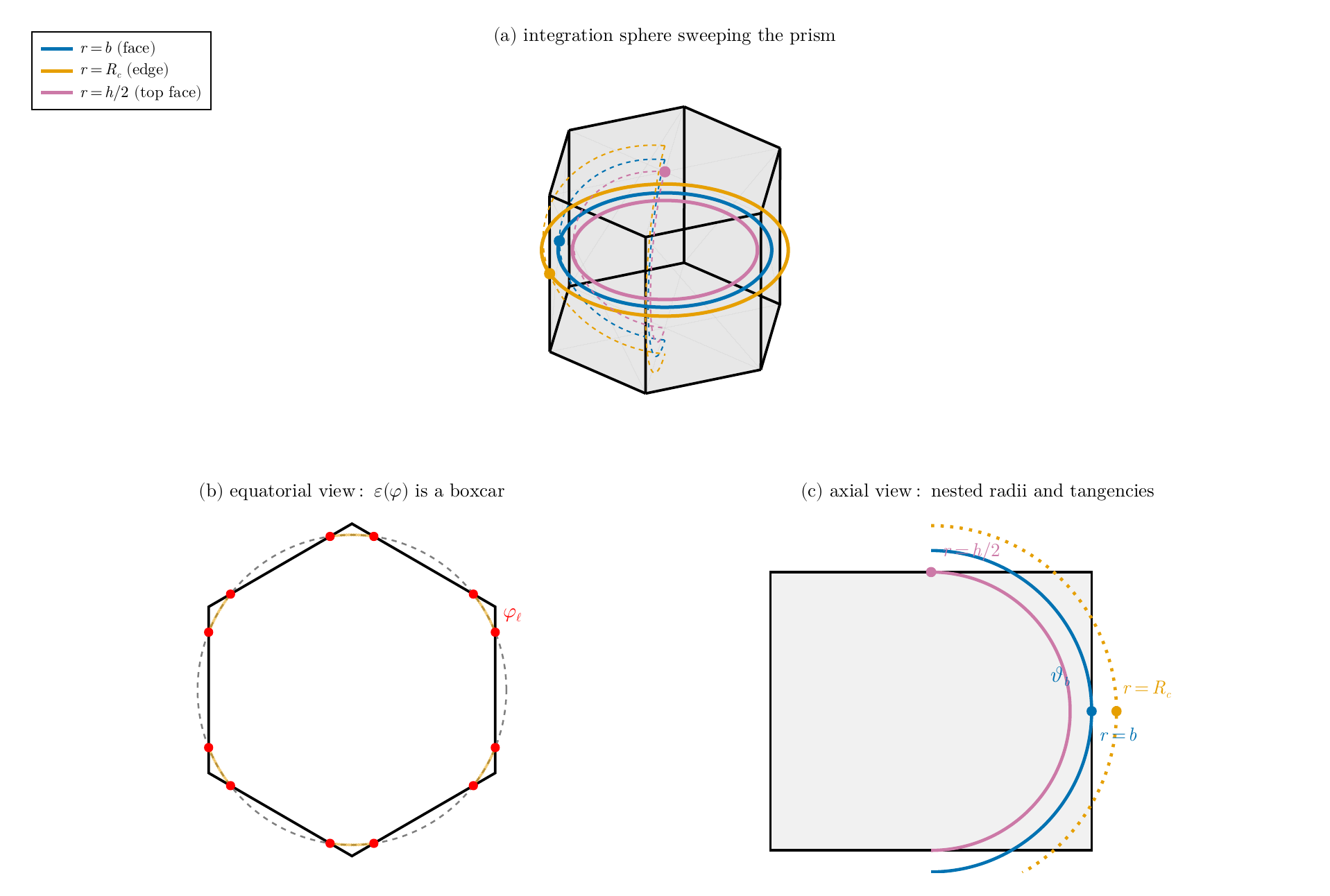}
  \caption{Geometry of the three IITM quadratures for a regular hexagonal prism
  (apothem $b$, circumradius $R_c=b/\cos(\pi/N)$, half-height $h/2$).
  \textbf{(a)} As the integration sphere of radius $r$ sweeps outward it becomes
  tangent to the particle at three feature types (the vertical faces ($r=b$), the
  vertical edges ($r=R_c$), and the top/bottom faces ($r=h/2$)), and each tangency
  type seeds one or more of the non-smooth branches analyzed below (e.g.\ the vertical
  face seeds both the zenithal $\sqrt{\,}$ branch and a radial kink).
  \textbf{(b)} In the equatorial plane, for $b<r_0<R_c$ the shell contrast
  $\varepsilon(\varphi)$ is a boxcar (unity inside the particle, near the vertices; zero
  outside, near the face centers) with jumps at the crossings $\varphi_\ell$ where
  $r_0=R(\varphi)$, the staircasing removed exactly by the closed-form arc
  coefficients (Lemma~\ref{lem:arc}).
  \textbf{(c)} In an axial plane the three radii nest: the inscribed-cylinder (face)
  tangency fixes the zenithal square-root breakpoint $\vartheta_b=\arcsin(b/r)$
  (Lemma~\ref{lem:sqrt}), while the edge tangency at $r=R_c$ (reached out of this plane,
  dotted) produces the $(r-R_c)^{3/2}$ radial branch (Proposition~\ref{prop:radial}).}
  \label{fig:schematic}
\end{figure}

\section{Analysis}
\label{sec:branches}

We now make the convergence statements precise. Throughout, ``analytic'' means
real-analytic with a complex-analytic continuation to a neighborhood of the real
interval (so that a Bernstein ellipse exists), and a regular $N$-gon cross-section
has apothem (inradius) $b$ and circumradius $R_c=b/\cos(\pi/N)$.

\subsection{Azimuthal exactness}
\label{sec:azimuth}

On a shell at cylindrical radius $\rho$ with $b<\rho<R_c$, the contrast is a boxcar
in $\varphi$: inside the polygon for $\varphi$ in a union of arcs centered on the
vertices, of half-width $w=\pi/N-\delta(\rho)$, $\delta(\rho)=\acos(b/\rho)$.

\begin{lemma}[Exact azimuthal coefficients]\label{lem:arc}
With the one-period, unnormalized convention
$c_q(\rho)=\int_0^{2\pi/N}[\varepsilon(\rho,\varphi)-1]\,e^{+\mathrm i q\varphi}\,\mathrm d\varphi$
(only $q\in N\mathbb Z$ survive for an $N$-fold particle), the coefficient is, for $q\neq0$,
\begin{equation}
  \label{eq:arc}
  c_q(\rho) \;=\; (m^2-1)\sum_{\mathrm{arcs}}
              \frac{e^{\mathrm{i}q\varphi_1}-e^{\mathrm{i}q\varphi_0}}{\mathrm{i}q},
\end{equation}
with $\varphi_{0,1}$ the boundary $\varphi$-crossings, and $c_0=(m^2-1)\sum_{\mathrm{arcs}}
(\varphi_1-\varphi_0)$. The expression is exact (no quadrature).
\end{lemma}
\noindent Equation \eqref{eq:arc} is the integral of a piecewise-constant function, so it
is exact by inspection. By contrast, the equidistant-sample-then-FFT coefficients of a
boxcar converge to \eqref{eq:arc} only at first order:
\begin{proposition}[Sharpness of the FFT rate]\label{prop:fft}
The $N_\varphi$-point equidistant (trapezoidal) Fourier coefficient $\hat c_q^{N_\varphi}$
of a function with a jump discontinuity satisfies the upper bound
$\hat c_q^{N_\varphi}-c_q = \Order(1/N_\varphi)$, and for boundary positions $\varphi_\ell$
in generic position the rate is attained along a subsequence (the error oscillates with
$N_\varphi$ via $\sin(N_\varphi\varphi_\ell)$, so it is no faster than $\Order(1/N_\varphi)$ but
not monotone).
\end{proposition}
\begin{proof}
Both bounds are textbook. The upper bound is Koksma's inequality \citep{davis1984} applied to the
bounded-variation integrand $(\varepsilon(\varphi)-1)\,e^{\mathrm{i}q\varphi}$ (variation
$\Order(q)$ at fixed $q$). Sharpness follows from the aliasing identity
$\hat c_q^{N_\varphi}=\sum_{j}c_{q+jN_\varphi}$ with the closed-form boxcar coefficients
of Lemma~\ref{lem:arc}: the leading aliases $c_{q\pm N_\varphi}=\Order(1/N_\varphi)$ carry the
phases $e^{\pm\mathrm{i}N_\varphi\varphi_\ell}$, which do not cancel for generic $\varphi_\ell$,
so the error oscillates as $\sum_\ell\beta_\ell\sin(N_\varphi\varphi_\ell)/N_\varphi$ (the source
of the oscillation in the measured exponent) and improves only when a boundary lands on a node, the
azimuthal analogue of placing a panel breakpoint at a discontinuity.
\end{proof}

\subsection{The square-root structure}
\label{sec:zenith}

At fixed shell radius $r$, $\rho=r\sin\vartheta$ and the inscribed-cylinder boundary
crossing is at $\vartheta_b(r)=\arcsin(b/r)$. The half-width $\delta$ has an exact
form that exposes its branch.

\begin{lemma}[Half-width identity and analyticity]\label{lem:sqrt}
Write $x=1-b/\rho$. Then
\begin{equation}
  \label{eq:acosid}
  \delta=\acos(b/\rho)=\acos(1-x)=2\arcsin\!\sqrt{x/2}.
\end{equation}
Consequently $\delta=\sqrt{2x}\,\Phi(x)$ with $\Phi$ analytic and $\Phi(0)=1$; and since
$x(\vartheta)=1-b/(r\sin\vartheta)$ is analytic with a simple zero at $\vartheta_b$, the
half-width $\delta$, and hence every coefficient $c_q$ in \eqref{eq:arc}, is an analytic
function of $s:=\sqrt{\vartheta-\vartheta_b}$.
\end{lemma}
\begin{proof}
With $u=1-t$ and then $u=v^2$,
\[
  \begin{aligned}
  \acos(1-x)
   &=\int_{1-x}^{1}\frac{dt}{\sqrt{1-t^2}}
    =\int_{0}^{x}\frac{du}{\sqrt{u(2-u)}}\\
   &=2\int_{0}^{\sqrt x}\frac{dv}{\sqrt{2-v^2}}
    =2\arcsin\!\sqrt{x/2},
  \end{aligned}
\]
which is \eqref{eq:acosid}. Since $\arcsin(\zeta)/\zeta$ is even and analytic, we have
$\arcsin\sqrt{x/2}=\sqrt{x/2}\,\Psi(x)$ with $\Psi$ analytic, giving $\delta=\sqrt{2x}\,\Phi(x)$. Now
$x(\vartheta)$ is analytic and $x(\vartheta_b)=0$ with
$x'(\vartheta_b)=b\cos\vartheta_b/(r\sin^2\vartheta_b)\neq0$, so $x=s^2\,\chi(s^2)$ with
$\chi$ analytic and $\chi(0)=x'(\vartheta_b)\neq0$; hence $\delta=s\,\widetilde\Phi(s^2)$
is analytic in $s$. Finally, with the inside half-width $w=\pi/N-\delta$,
$c_q=(m^2-1)\,e^{\mathrm{i}q\pi/N}\,2\sin(qw)/q
=-(m^2-1)\,e^{\mathrm{i}q\pi/N}\cos(q\pi/N)\,2\sin(q\delta)/q$ for
$q\neq0$ (using $\sin(q\pi/N)=0$ for $q\in N\mathbb Z$), and $c_0=(m^2-1)(2\pi/N-2\delta)$;
$\sin(q\delta)$ is an analytic
function of the analytic $\delta(s)$, so $c_q$ is analytic in $s$.
\end{proof}

Spectral convergence requires analyticity on \emph{every} panel, so we first classify
all zenithal breakpoints; unlike the radial case (Proposition~\ref{prop:radial} below), here
$c_q(\vartheta)$ is known in closed form, so the classification is exhaustive.

\begin{lemma}[Zenithal breakpoints]\label{lem:zbreak}
At fixed $r$, $c_q(\vartheta)$ is real-analytic on $(0,\pi/2)$ except at the breakpoints
in $\{\vartheta_{\mathrm{cap}},\vartheta_b,\vartheta_{R_c}\}\cap(0,\pi/2)$,
\[
  \vartheta_{\mathrm{cap}}=\acos\!\big(\tfrac{h}{2r}\big),\qquad
  \vartheta_b=\arcsin(b/r),\qquad
  \vartheta_{R_c}=\arcsin(R_c/r),
\]
where it has respectively: \emph{(i)} a jump (the
slice crosses the cap plane $z=\pm h/2$); \emph{(ii)} a square-root branch (the
inscribed-cylinder tangency, Lemma~\ref{lem:sqrt}); \emph{(iii)} an analytic simple
zero (the circumradius crossing $\vartheta_{R_c}$, a \emph{zenithal} $C^0$ event, distinct from
the \emph{radial} edge tangency $r=R_c$ that produces the $(r-R_c)^{3/2}$ branch of
Proposition~\ref{prop:radial}). On each side of every breakpoint, and up to the
endpoints $0,\pi/2$, $c_q$ extends analytically (at $\vartheta_{R_c}$ the two one-sided
continuations differ ($\propto\sin(q\delta)$ from below, $\equiv0$ from above) but
share their value (both vanish) at the simple zero, a $C^0$ kink, so each one-sided panel's
integrand is analytic up to that endpoint, which is all Theorem~\ref{thm:zenith} requires).
\end{lemma}
\begin{proof}
For $b<\rho<R_c$, $|z|<h/2$, one has $b/\rho\in(0,1)$, so $\delta=\acos(b/\rho)$
is analytic in $\vartheta$ and $c_q$ (a finite combination of $\sin(q\delta)$ and
constants, $q$ a multiple of $N$) is analytic. The slab constraint $|z|\le h/2$ is
independent of $\varphi$ ($z=r\cos\vartheta$), so crossing $z=\pm h/2$ gates the entire
slice on or off: $c_q$ jumps from $0$ to the lateral-arc value $c_q(\rho(\vartheta_{\mathrm{cap}}))$,
with no $\varphi$-structure contributed by the cap face itself, $\Rightarrow$ (i). Next,
(ii) is Lemma~\ref{lem:sqrt}, the only
locus where $b/\rho\to1$; at $\rho=R_c$, $b/\rho=\cos(\pi/N)$ is bounded away from $1$, so
$\delta$ remains analytic and $\sin(q\delta)\to\sin(q\pi/N)=0$ vanishes linearly,
giving (iii).
\end{proof}

\begin{theorem}[Spectral zenithal convergence]\label{thm:zenith}
Split $[0,\pi/2]$ at the breakpoints of Lemma~\ref{lem:zbreak} and apply the square-root
substitution
\begin{equation}
  \label{eq:sub}
  x=x_c+t^2,\qquad dx=2t\,dt
\end{equation}
(here $x=\vartheta$, $x_c=\vartheta_b$) on the inscribed-cylinder panel. On this mirror-folded
half-range each partial-arc panel carries a \emph{single} square-root endpoint, since the equator
$\vartheta=\pi/2$ is a symmetry boundary, not a branch, so one substitution per panel suffices.
(In the full-polar case of a particle without a horizontal mirror plane, a panel can be bounded by
two square-root tangencies, which needs the additional midpoint bisection of
Remark~\ref{rem:poly}.) Then on \emph{every} resulting panel the integrand is
analytic up to both endpoints, so it admits a Bernstein ellipse $E_\varrho$,
$\varrho>1$, and the $n$-point Gauss--Legendre error is $\Order(\varrho^{-2n})$.
\end{theorem}
\begin{proof}
By Lemma~\ref{lem:zbreak} each interior breakpoint is a jump, a square-root branch, or
an analytic zero, and splitting places each at a panel endpoint (the equator $\pi/2$ bounding the
last partial panel as a symmetry endpoint). At a jump or
analytic-zero endpoint $c_q$ already extends analytically (Lemma~\ref{lem:zbreak}); at
the square-root endpoint $\vartheta_b$ the substitution \eqref{eq:sub} makes $c_q$
analytic in $t$ (Lemma~\ref{lem:sqrt}). The weights $g_{pp'}$ (the associated
Legendre/trigonometric products of \eqref{eq:Uzenith}, analytic on each panel; any apparent pole
at $\vartheta=0,\pi$ is canceled by the $\sin^{|m|}\vartheta$ factors) and $\sin\vartheta$ are analytic in $t=
\sqrt{\vartheta-\vartheta_b}$ (or in $\vartheta$ on unsubstituted panels), and the
Jacobian $2t$ is entire. Hence on every panel the integrand is analytic on a complex
neighborhood of the closed interval, a Bernstein ellipse exists, and the classical
Gauss bound $\Order(\varrho^{-2n})$ applies \citep{trefethen2013,davis1984}.
\end{proof}

\begin{remark}[Uniformity in $r$ and $q$]\label{rem:unif}
The rate $\varrho$ and the constants are not claimed uniform: $\varrho\to1$ as a panel
collapses (e.g.\ $r\to b$, where $\vartheta_b\to\pi/2$) or as $q$ grows, and the
bounded-variation constant of Proposition~\ref{prop:fft} grows like $q$. The multipole
truncation bounds $q\le 2n_{\max}$, and away from the finitely many radii where two
breakpoints merge the panel lengths are bounded below; a bound uniform across the entire
radial quadrature is not pursued here.
\end{remark}

Without the substitution, the same split leaves a square-root endpoint and the rate is
only algebraic:
\begin{proposition}[Sharpness of the unsubstituted rate]\label{prop:gauss}
For $\int_0^c x^\alpha g(x)\,dx$ with $g$ analytic, $\alpha>-1$, and
$\alpha\notin\mathbb{Z}_{\ge0}$,
$n$-point Gauss--Legendre converges at $\Order(n^{-2\alpha-2})$
\citep{davis1984,trefethen2013}. For the unsubstituted inscribed-cylinder panel the
$\sqrt{\vartheta-\vartheta_b}$ endpoint gives $\alpha=\tfrac12$, hence $\Order(n^{-3})$.
\end{proposition}

\subsection{Geometric classification of the radial branches}
\label{sec:radial}

After exact angular integration the radial integrand of \eqref{eq:riccati} is
non-smooth at the radii $r_c$ where the integration sphere is tangent to a feature of
the particle. The branch exponent is fixed by the \emph{type} of the tangency, through
the local extremum of the polar boundary $R(\varphi)$ at the tangency azimuth.

\begin{proposition}[Branch classification of a contact tangency]\label{prop:radial}
Assume that at $\xi=0$ the integration sphere $r=r_c(1+\xi)$ makes a
\emph{non-degenerate} contact with the particle: a transversal tangency to one face or one
straight edge at an isolated point, or to one smooth boundary curve along which the local model
is uniform (e.g.\ the spheroid equator or the cylinder wall). (Several such contacts may share
the same radius, as the symmetric edges and faces of an $N$-fold particle do, in which
case their contributions simply add and the statement below characterizes each. Excluded are
only the degenerate coincidences where two loci \emph{merge at a single point}, vertices, and
non-transversal contacts; these form the higher-codimension case of \S\ref{sec:discussion} and
are not analyzed here.) In local coordinates
$\eta=\tfrac\pi2-\vartheta$ (the polar deviation from the equator) and $\varphi$ about the
tangency point, the angular integral obeys
\begin{equation}\label{eq:branchexp}
  I(r)=\mathcal A(\xi)+C\,\xi^{\sigma}+o(\xi^{\sigma}),
  \qquad \mathcal A\ \text{analytic at }0,
\end{equation}
where the exponent $\sigma$ is fixed by the local geometry of the polar boundary
$R(\varphi)$ at the tangency azimuth,
\[
  \sigma=\begin{cases}
    3/2 & \text{straight edge (prism vertical or top/bottom rim)},\\
    1   & \text{face, or a convex edge along a curve (linear pinch)},\\
    1/2 & \text{\emph{smooth} (osculating) tangency along a curve},
  \end{cases}
\]
and the leading coefficient $C$ is a \emph{local invariant} of the tangency
(determined by the corner slope and the value of the integrand there); for the
edge case $C=+\tfrac{8\sqrt2}{3\tau}\,w_0$ for the retained (inside) integral $I$, with
$\tau=\tan(\pi/N)$ and $w_0$ the integrand at the equatorial tangency (the sign flips for the
excluded integral). $C\neq0$ for generic observables (it may vanish for individual matrix
entries by symmetry).
\end{proposition}
\begin{proof}[Sketch]
Each exponent follows from the excluded angular measure in a rescaled inner region. A straight
edge has a linear corner $R(\varphi)=R_c(1-\tau|\varphi|)$, whose excluded measure
$\tfrac{8\sqrt2}{3\tau}[(\xi+\tau\varphi_0)^{3/2}-\xi^{3/2}]$ has non-analytic part
$-\tfrac{8\sqrt2}{3\tau}\xi^{3/2}$, giving $\sigma=\tfrac32$ and
$|C|=\tfrac{8\sqrt2}{3\tau}|w_0|$; a lateral face is a quadratic minimum, with excluded set an
ellipse of measure $\pi\sqrt{2/c}\,\xi$ ($\sigma=1$, a slope kink removed by splitting at
$r_c$); a smooth (osculating) tangency along a curve, the spheroid equator, gives an
equatorial band of measure $\Order(\xi^{1/2})$ ($\sigma=\tfrac12$, the $\sqrt{r-\rmin}$
branch). The rescaling that makes $C$ exact and bounds the $o(\xi^\sigma)$ remainder
(Remark~\ref{rem:status}) is given in \ref{app:branches}.
\end{proof}

\begin{remark}[Precise status]\label{rem:status}
Proposition~\ref{prop:radial} establishes the \emph{leading} non-analytic term: the
exponent $\sigma$ and the coefficient $C$ are exact. It does \emph{not} assert the full
non-analytic expansion (subleading terms, generically $\xi^{5/2}$, with a
possible $\xi^2\log\xi$ from inner/outer matching, are not computed),
and the $o(\xi^\sigma)$ remainder is controlled by the rescaling above rather than by
a fully written uniform bound. Only the leading term is needed here: it fixes both the
order cap $\min(\sigma\!+\!1,\,\text{scheme order})$ of any radial integrator and the
half-integer $\sigma$ that the substitution removes.
\end{remark}

All three exponents are integer or half-integer; the substitution \eqref{eq:sub},
$r=r_c+t^2$, maps $(r-r_c)^{k/2}\mapsto t^{k}$ and removes the branch. The face kink is
already removed by splitting at $r_c$.

\subsection{Radial integrator order}
\label{sec:radord}

The order cap is not the scalar-quadrature endpoint result but a statement about a
linear ODE \emph{system} with a branch in its coefficient, followed by the rational
reconstruction $\Tmat=\mathbf V\mathbf P^{-1}$; we isolate it.

\begin{lemma}[Order reduction at a coefficient branch]\label{lem:orderred}
Let $\tfrac{d}{dr}\mathbf y=\mathcal M(r)\mathbf y$ on $[r_c,r_c+\Lambda]$ with
$\mathcal M(r)=\mathcal M_0(r)+(r-r_c)^\sigma\mathcal M_1(r)$, $\mathcal M_0,\mathcal M_1$
analytic and $\sigma>0$ non-integer. A one-step method of classical order $p$ on a
uniform grid of spacing $h$ with $r_c$ as a node has global error
$\Order(h^{\min(p,\,\sigma+1)})$. If the represented quantity is reconstructed as
$\Tmat=\mathbf V\mathbf P^{-1}$ with $\mathbf P$ uniformly invertible on the trajectory,
the same order holds for $\Tmat$.
\end{lemma}
\begin{proof}[Sketch]
Integrating once, $\mathbf y=\mathbf y_{\mathrm{an}}+(r-r_c)^{\sigma+1}\mathbf w$ with
$\mathbf y_{\mathrm{an}}$ analytic, so $\mathbf y^{(k)}=\Order((r-r_c)^{\sigma+1-k})$; summing the
one-step local errors $\Order(h^{p+1}(jh)^{\sigma-p})$ over the grid gives global
$\Order(h^{\min(p,\sigma+1)})$, and the analytic reconstruction $\Tmat=\mathbf V\mathbf P^{-1}$
transfers the order to $\Tmat$. The summation, the borderline cases, and the reconstruction bound
are in \ref{app:radial}.
\end{proof}

\begin{proposition}\label{prop:radord}
The fixed-coefficient product recursion (Johnson/Bi) for \eqref{eq:riccati} has
global order $1$ (proved below); more generally, discretizations that preserve the
positivity/monotonicity structure of the Riccati flow are constrained to low order
(cf.\ \citet{dieci1996}). The linear-lift (M\"obius)
integrator $\Tmat=\mathbf V\mathbf P^{-1}$,
$\tfrac{d}{dr}[\mathbf P;\mathbf V]=\mathcal M(r)[\mathbf P;\mathbf V]$ with a
$p$-th order one-step method on the linear system has order $p$ for analytic
$\mathcal M$, but a $(r-r_c)^{\sigma}$ branch in $\mathcal M$ caps the realized order
at $\min(p,\sigma+1)$. Substitution \eqref{eq:sub} removes the leading half-integer $\sigma$,
restoring the scheme's design order, here the fourth order of the fourth-order Runge--Kutta
(RK4) lift used below, provided no
lower-regularity subleading term remains (Remark~\ref{rem:status}).
\end{proposition}
\noindent The product recursion holds $\Umat$ fixed at the node and resums intra-shell
scattering through a resolvent; its local truncation error is $\Order(\Delta r^2)$
(coefficient held fixed $+$ operator non-commutativity), hence global $\Order(\Delta r)$.
The lift is a linear ODE, on which a $p$-th order Runge--Kutta or Magnus/Pad\'e step
\citep{schiff1999,hairer2006,doicu2019} has order $p$; the order reduction at a coefficient branch
$(r-r_c)^\sigma$ in $\mathcal M$ is Lemma~\ref{lem:orderred}, and \eqref{eq:sub} is its
remedy: mapping the \emph{leading} half-integer $\sigma$ to an integer power lifts the
coefficient regularity above the design order. For the fourth-order RK4 lift used here this
restores fourth order; a higher-order scheme would in addition require the lower-regularity
subleading terms to be absent (the possible $t^4\log t$ from inner/outer matching is discussed in
\ref{app:radial}). Here
$\sigma$ is the same branch exponent as in
Proposition~\ref{prop:radial}, since $\mathcal M\propto\Umat$ inherits the branch of the
angular integral $I(r)$; the half-integer cases ($\sigma=\tfrac12,\tfrac32$) are the
ones Lemma~\ref{lem:orderred} governs, whereas the integer face case $\sigma=1$ is an
ordinary kink removed by splitting the radial panel at $r_c$ (outside the
non-integer hypothesis of Lemma~\ref{lem:orderred}). The
M\"obius update $\Tmat\leftarrow(\Phi_{21}+\Phi_{22}\Tmat)(\Phi_{11}+\Phi_{12}\Tmat)^{-1}$
is well-defined while $\Phi_{11}+\Phi_{12}\Tmat$ stays nonsingular, i.e.\ away from
conjugate points of the Hamiltonian flow; for the physical T-matrix these occur only at
internal resonances, and per-step renormalization keeps the propagated subspace
well-conditioned in between. At large size parameter the internal resonances, hence the
conjugate points, grow denser, so a discrete step can approach a near-singular update even
away from a physical resonance; we do not characterize this robustness at large $k\rmax$ here.

\subsection{Generalization to convex polyhedra}
\label{sec:polyhedra}

The prism analysis used only the inradius $b$ and circumradius $R_c$, but the three
ingredients (closed-form arcs, zenithal breakpoints, and critical radii) are
properties of the \emph{contact geometry} and extend to any convex polyhedron. Write the
scatterer as an intersection of face half-spaces
$\Omega=\bigcap_f\{x:\mathbf n_f\!\cdot x\le d_f\}$, with outward unit normals
$\mathbf n_f=(\sin\alpha_f\cos\beta_f,\sin\alpha_f\sin\beta_f,\cos\alpha_f)$ and plane
offsets $d_f\ge0$. For a shell point $x=r\,\hat u(\vartheta,\varphi)$ the half-space
condition $\mathbf n_f\!\cdot x\le d_f$ becomes
\begin{equation}\label{eq:facecond}
  A_f(\vartheta)\,\cos(\varphi-\beta_f)\le B_f(r,\vartheta),\quad
  A_f=\sin\alpha_f\sin\vartheta,\quad B_f=\frac{d_f}{r}-\cos\alpha_f\cos\vartheta .
\end{equation}

\begin{proposition}[Convex-polyhedron breakpoints]\label{prop:poly}
Let $\Omega$ be a convex polyhedron, and assume each contributing face meets the integration
sphere transversally, with no sphere passing through a vertex or tangent to two faces \emph{at a
common point} (separated tangencies sharing the same radius, as in an $N$-fold particle,
are admissible and add; the excluded coincidences are the higher-codimension case of
\S\ref{sec:discussion}, not analyzed here). With \eqref{eq:facecond}:
\emph{(i) Azimuth.} A \emph{non-axial} face ($A_f\neq0$, i.e.\ $\alpha_f\notin\{0,\pi\}$) with
$|B_f/A_f|<1$ excludes the $\varphi$-arc of half-width $\psi_f=\acos(B_f/A_f)$ centered on
$\varphi=\beta_f$; the shell-circle interior is the complement, a finite union of arcs, with the
closed-form arc coefficients \eqref{eq:arc} over those arcs (Lemma~\ref{lem:arc} is the single
regular-prism case). An \emph{axial} face ($A_f=0$: a cap, $\alpha_f\in\{0,\pi\}$) imposes no
$\varphi$-constraint; it gates the whole circle on or off through the $z$-condition
$\cos\alpha_f\cos\vartheta\le d_f/r$, contributing the jump of (ii) rather than an arc.
\emph{(ii) Zenith.} At fixed $r$, $c_q(\vartheta)$ is real-analytic except at the
breakpoints formed by the face tangencies $\vartheta=\pm\alpha_f\pm\acos(d_f/r)$ and the
edge--sphere crossings (each edge line $\{\mathbf n_f\!\cdot x=d_f,\ \mathbf n_g\!\cdot
x=d_g\}$ met with $|x|=r$), retained only when the witness point lies on $\partial\Omega$;
at each it has a jump, a square-root branch, or an analytic zero as classified in
Lemma~\ref{lem:zbreak}. \emph{(iii) Radius.} The critical radii are $r=d_f$ (face-plane
tangency), $r=|v|$ (vertices $v$), and $r=\operatorname{dist}(O,\text{edge})$ (edge-line
tangencies). The branch exponents are unchanged from Proposition~\ref{prop:radial}: a face
contributes $\sigma=1$, a straight edge $\sigma=\tfrac32$.
\end{proposition}
\begin{proof}[Sketch]
Substituting $\hat u$ into $\mathbf n_f\!\cdot x\le d_f$ gives \eqref{eq:facecond}; the surviving
$\varphi$-set is the complement of the excluded arc $|\varphi-\beta_f|<\acos(B_f/A_f)$ (non-axial
face) or the whole circle gated by the $z$-condition $\cos\alpha_f\cos\vartheta\le d_f/r$ (axial
face), the zenithal tangencies sit where the circle's extreme value $\cos(\vartheta\mp\alpha_f)$
equals $d_f/r$, i.e.\ $\vartheta=\pm\alpha_f\pm\acos(d_f/r)$, and the critical radii are $r=d_f$,
$r=|v|$, and the edge-line distances; the branch types follow from the contact order as in
Lemma~\ref{lem:zbreak} and Proposition~\ref{prop:radial}. The full derivation, including the
on-boundary witness test that removes spurious face-plane tangencies, is in \ref{app:poly}.
\end{proof}

These loci reduce to the regular-prism breakpoints of \S\ref{sec:radial} for axial caps and
vertical sides, and we verified them against brute-force point-in-polyhedron tests on the
solid bullet. One subtlety (an equatorial-belt panel bounded by a square-root tangency at
\emph{both} ends, which a single substitution \eqref{eq:sub} cannot absorb and which we
therefore split at its midpoint before substituting) is recorded with these checks in
\ref{app:poly} (Remark~\ref{rem:poly}).

\section{Boundary-conformal integration scheme}
\label{sec:scheme}

The scheme is the conjunction of three ingredients, all driven by the
analytically known tangency loci:
\begin{enumerate}
\item \textbf{Azimuth:} replace the FFT by the closed-form arc coefficients
  $c_q$ \eqref{eq:arc} (Proposition~\ref{prop:poly}(i) for a general convex polyhedron).
\item \textbf{Zenith:} split the polar range at the breakpoints
  (Lemma~\ref{lem:zbreak}; Proposition~\ref{prop:poly}(ii)) and apply the substitution
  \eqref{eq:sub} on every $\sqrt{\,}$-branch panel. A particle with a horizontal mirror
  plane (prism, cylinder, spheroid) is integrated on $[0,\pi/2]$, the two hemispheres folded
  together by the up--down mirror (parity) symmetry; a particle without it (the bullet) is
  integrated on the full $[0,\pi]$
  with all polar breakpoints enumerated, including the equatorial double-branch split
  (Remark~\ref{rem:poly}).
\item \textbf{Radius:} integrate \eqref{eq:riccati} by the structure-preserving
  (M\"obius / linear-lift) scheme of Algorithm~\ref{alg:march} (following
  \citet{doicu2019,schiff1999}) on panels split at the critical radii, with the
  $t^2$ substitution at each half-integer critical radius.
\end{enumerate}

For the radius we integrate the \emph{linear lift} of the Riccati flow,
\begin{equation}
  \label{eq:lift}
  \Tmat=\mathbf V\mathbf P^{-1},\qquad
  \frac{d}{dr}\begin{bmatrix}\mathbf P\\\mathbf V\end{bmatrix}
    =\mathcal M(r)\begin{bmatrix}\mathbf P\\\mathbf V\end{bmatrix},\qquad
  \mathcal M=\mathrm{i}k\bigl[-\mathbf H^{\!\top};\,\mathbf J^{\!\top}\bigr]\,
            \Umat\,\bigl[\mathbf J\ \ \mathbf H\bigr],
\end{equation}
by the classical fourth-order Runge--Kutta (RK4) method introduced above. With $L$ the number of VSWF modes in the
symmetry block, the per-block T-matrix is $2L\times2L$ (two polarizations per mode) and its lift
$[\mathbf P;\mathbf V]$ propagates by a $4L\times4L$ matrix. Over each substep $[r,r+h]$ we form
the $4L\times4L$ step propagator $\Phi=\mathbf I+\tfrac h6(K_1+2K_2+2K_3+K_4)$ from the RK4 stages
$K_i$ of $\mathbf Y'=\mathcal M\mathbf Y$, $\mathbf Y(r)=\mathbf I$ (three evaluations of
$\mathcal M$, at $r,\,r+\tfrac h2,\,r+h$), and advance $\Tmat$ by the M\"obius update
$\Tmat\leftarrow(\Phi_{21}+\Phi_{22}\Tmat)(\Phi_{11}+\Phi_{12}\Tmat)^{-1}$
(Algorithm~\ref{alg:march}). \emph{Per-step renormalization} means the following: propagating
$\Tmat$ itself by the M\"obius map is equivalent to restarting the lift with
$\mathbf P=\mathbf I,\ \mathbf V=\Tmat$ at the start of every step, so we never form
$\mathbf V\mathbf P^{-1}$ with an exponentially ill-conditioned $\mathbf P$. The update is
well defined while $\Phi_{11}+\Phi_{12}\Tmat$ is nonsingular (away from conjugate points of the
Hamiltonian flow, i.e.\ internal resonances). This is stable where direct propagation of the
exponentially growing Neumann solutions is not. That a direct, positivity-preserving Riccati
recursion is constrained to low order parallels the \citet{dieci1996} monotonicity obstruction;
the lift circumvents it.

\begin{algorithm}[htbp]
\caption{Boundary-conformal IITM radial march (one $N$-fold symmetry block).}
\label{alg:march}
\begin{algorithmic}[1]
\State enumerate the critical radii $\rmin=r_0<\dots<r_M=\rmax$
       (Proposition~\ref{prop:poly}(iii)); set $\Tmat\gets\Tmat_{\mathrm{Mie}}(\rmin)$
\For{each radial panel $[r_j,r_{j+1}]$, $t^2$-mapped if an endpoint is a half-integer critical radius}
  \For{each uniform RK4 substep $[r,r+h]$ of the panel}
    \State assemble $\mathcal M_a=\mathcal M(r)$, $\mathcal M_b=\mathcal M(r+\tfrac h2)$,
           $\mathcal M_c=\mathcal M(r+h)$ from $\Umat$ \eqref{eq:lift} via the conformal angular
           quadrature \eqref{eq:Uzenith} (arc coefficients \eqref{eq:arc}; zenith split $+\,t^2$)
    \State $K_1\!=\!\mathcal M_a$, $K_2\!=\!\mathcal M_b(\mathbf I\!+\!\tfrac h2 K_1)$,
           $K_3\!=\!\mathcal M_b(\mathbf I\!+\!\tfrac h2 K_2)$, $K_4\!=\!\mathcal M_c(\mathbf I\!+\!hK_3)$;\;\;
           $\Phi\gets\mathbf I+\tfrac h6(K_1+2K_2+2K_3+K_4)$
           \Comment{RK4 for $\mathbf Y'=\mathcal M\mathbf Y$, $\mathbf Y(r)=\mathbf I$, over the $4L\times4L$ block}
    \State $\Tmat\gets(\Phi_{21}+\Phi_{22}\Tmat)(\Phi_{11}+\Phi_{12}\Tmat)^{-1}$
           \Comment{M\"obius update $\equiv$ reset $\mathbf P\!=\!\mathbf I,\,\mathbf V\!=\!\Tmat$}
  \EndFor
\EndFor
\end{algorithmic}
\end{algorithm}

\section{Numerical results}
\label{sec:results}

All experiments use the hexagonal prism ($R_c=1$, apothem $b=\tfrac{\sqrt3}{2}$,
height $h=1.5$, refractive index $m=1.5$, $k\rmax=1.25$); the prolate spheroid with equatorial semi-axis $1$ and polar semi-axis
$1.6$ (the smooth-radial control), and the finite cylinder of radius $a=1$ and height
$h=1$ ($m=1.5$; the faceted axisymmetric case). The size parameter is kept modest so the
reference solutions reach machine precision cheaply; the convergence \emph{orders}
reported here are fixed by the local contact geometry (Theorem~\ref{thm:zenith} and
Proposition~\ref{prop:radial}) and are $k\rmax$-independent over the range tested
($k\rmax\le20$, all three directions, Table~\ref{tab:krmax}). We therefore expect the same
\emph{convergence orders} on the large faceted particles that motivate the IITM, while the
resolution at which each direction enters its asymptotic regime
(Fig.~\ref{fig:zenithsize}) grows with size.
Convergence orders $p$ are least-squares fits of $\log(\text{error})$ versus $\log N$ over the
asymptotic tail of the tested resolution ladder: unless a figure states otherwise, the fit uses
the last four to five resolutions, and any pre-asymptotic low-resolution points are excluded (as
noted for the bullet in Fig.~\ref{fig:bullet}). Two distinct reference values are
used, and should not be confused: the \emph{order} figures measure each scheme's relative error
against a high-resolution sample of the \emph{same} scheme (whose own error is far below the
plotted range, so the measured slope is independent of this choice); the \emph{absolute}
cross-section value for the bullet (\S\ref{sec:results}, $C_{\mathrm{sca}}^\star$) is instead an
\emph{extrapolated} reference (Aitken/Richardson of the conformal sequence), cross-validated to
the independent plain IITM's own $\sim10^{-4}$ discretization floor. The angular and radial order figures (Figs.~\ref{fig:angular}--\ref{fig:radial})
are computed at reduced multipole truncations ($n_{\max}=4,6,8$ for the prism, spheroid, and
cylinder respectively) for speed; by the local regularity
analysis of \S\ref{sec:branches} the convergence \emph{order} of every direction is set by the
contact geometry and is $n_{\max}$-independent (Table~\ref{tab:krmax} verifies this directly up to
$k\rmax=20$, $n_{\max}=33$, in all three directions), so the reduced $n_{\max}$ affects only the
absolute error magnitudes, not the slopes. Each reported observable (here
$C_{\mathrm{sca}}$) is a smooth function of finitely many T-matrix entries, each
assembled from finitely many coefficients $c_q$ and one radial integration. Since none of
these finite operations improves the regularity of the least-regular factor, the observable's
convergence order is the minimum of the per-direction orders established in
\S\ref{sec:branches}, which is what the table records.

\begin{table}[htbp]
\centering
\caption{Convergence order of each integration direction, standard versus
boundary-conformal. ``exact'' denotes machine precision independent of the grid.
``Standard'' is the plain scheme with no breakpoint handling: an equidistant azimuth
rule, an unsplit zenithal Gauss rule (cap-jump dominated, $\Order(1/N)$), and the
first-order radial product recursion (Proposition~\ref{prop:radord}). Intermediate single-fix
schemes (zenithal split-only $\Order(N^{-3})$, Proposition~\ref{prop:gauss}; radial lift-only
$p\approx\confOrderPrismMobius$) appear in Figures~\ref{fig:angular}--\ref{fig:radial}.
All radial orders are least-squares fits; $^{\dagger}$ marks the
one entry whose fitted slope sits \emph{above} its theoretical order: the standard prism radial is
first order (Proposition~\ref{prop:radord}); the listed $\confOrderPrismProduct$ is a pre-asymptotic slope.}
\label{tab:orders}
{\small\setlength{\tabcolsep}{3pt}
\begin{tabular}{llll}
\toprule
Direction & Singularity at tangency & Standard & Boundary-conformal \\
\midrule
Azimuth ($N_\varphi$) & boxcar jump            & $\Order(N^{-1})$ & exact \\
Zenith ($N_\vartheta$) & $\sqrt{\vartheta-\vartheta_b}$ & $\Order(N^{-1})$ & spectral \\
Radius, prism ($N_r$) & kink $+\,(r-R_c)^{3/2}$ & $p\approx \confOrderPrismProduct^{\dagger}$ & $p\approx \confOrderFourth$ \\
Radius, spheroid ($N_r$) & $\sqrt{r-\rmin}$    & $p\approx \confOrderSpheroidNosub$ & $p\approx \confOrderFourth$ \\
Radius, cylinder ($N_r$) & $\sqrt{r-a}$ lateral & $p\approx \confOrderCylNosubOne$ & $p\approx \confOrderFourth$ \\
\bottomrule
\end{tabular}}
\end{table}

\begin{table}[htbp]
\centering
\caption{Per-direction boundary-conformal convergence for the hexagonal prism across size
parameter $k\rmax$ (the multipole truncation $n_{\max}$ grown with $k\rmax$ by the Mie rule
$n_{\max}\approx k\rmax+4(k\rmax)^{1/3}+2$). The azimuthal arc coefficients are exact at any size
by construction (Lemma~\ref{lem:arc}); the zenithal substitution stays spectral (relative error in
$C_{\mathrm{sca}}$ falling to $\sim$\num{e-10}); and the radial $t^2$ lift holds its fourth order
(Proposition~\ref{prop:radial}). At $k\rmax=20$ the coarsest radial nodes are still
pre-asymptotic (the full-ladder least-squares slope reads $\confKrmaxFullLadderSlope$, the asymptotic tail
($N_r\ge16$) $\approx4$), and the zenithal spectral onset likewise recedes with size
(Fig.~\ref{fig:zenithsize}, the high-$q$ regime of Remark~\ref{rem:unif}).}
\label{tab:krmax}
\begin{tabular}{lccccc}
\toprule
$k\rmax$ & $1.25$ & $2.5$ & $5$ & $10$ & $20$ \\
$n_{\max}$ & $8$ & $10$ & $14$ & $21$ & $33$ \\
\midrule
Azimuth & exact & exact & exact & exact & exact \\
Zenith & spectral & spectral & spectral & spectral & spectral \\
Radial order $p$ & $\confKrmaxOrderOnePtTwoFive$ & $\confKrmaxOrderTwoPtFive$ & $\confKrmaxOrderFive$ & $\confKrmaxOrderTen$ & ${\approx}4$ \\
\bottomrule
\end{tabular}
\end{table}

\begin{figure}[htbp]
  \centering
  \includegraphics[width=0.6\textwidth]{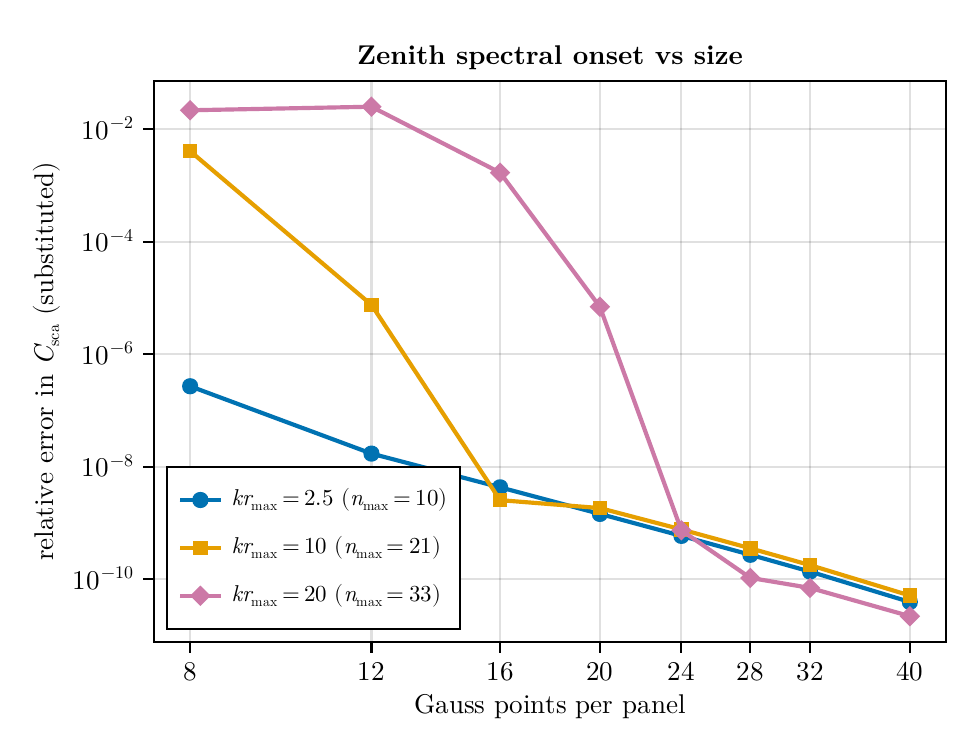}
  \caption{Size dependence of the substituted ($t^2$) zenithal quadrature (hexagonal prism). The
  convergence is spectral at \emph{every} size (all three curves reach the same
  $\sim$\num{e-10} floor), but the resolution at which it \emph{enters} the spectral regime
  recedes with $k\rmax$: the collapse moves from $\approx\confOnsetTwoPtFive$ Gauss points per panel at
  $k\rmax=2.5$ ($n_{\max}=10$), through $\approx\confOnsetTen$ at $k\rmax=10$ ($n_{\max}=21$), to
  $\approx\confOnsetTwenty$ at $k\rmax=20$ ($n_{\max}=33$), as the higher-$q$ Wigner-$d$ content
  (Remark~\ref{rem:unif}) demands more nodes. (``Onset'' = entry into the spectral regime, a
  coarser landmark than reaching a fixed error \emph{level} as in Fig.~\ref{fig:angular}.)}
  \label{fig:zenithsize}
\end{figure}

\begin{figure}[htbp]
  \centering
  \includegraphics[width=\textwidth]{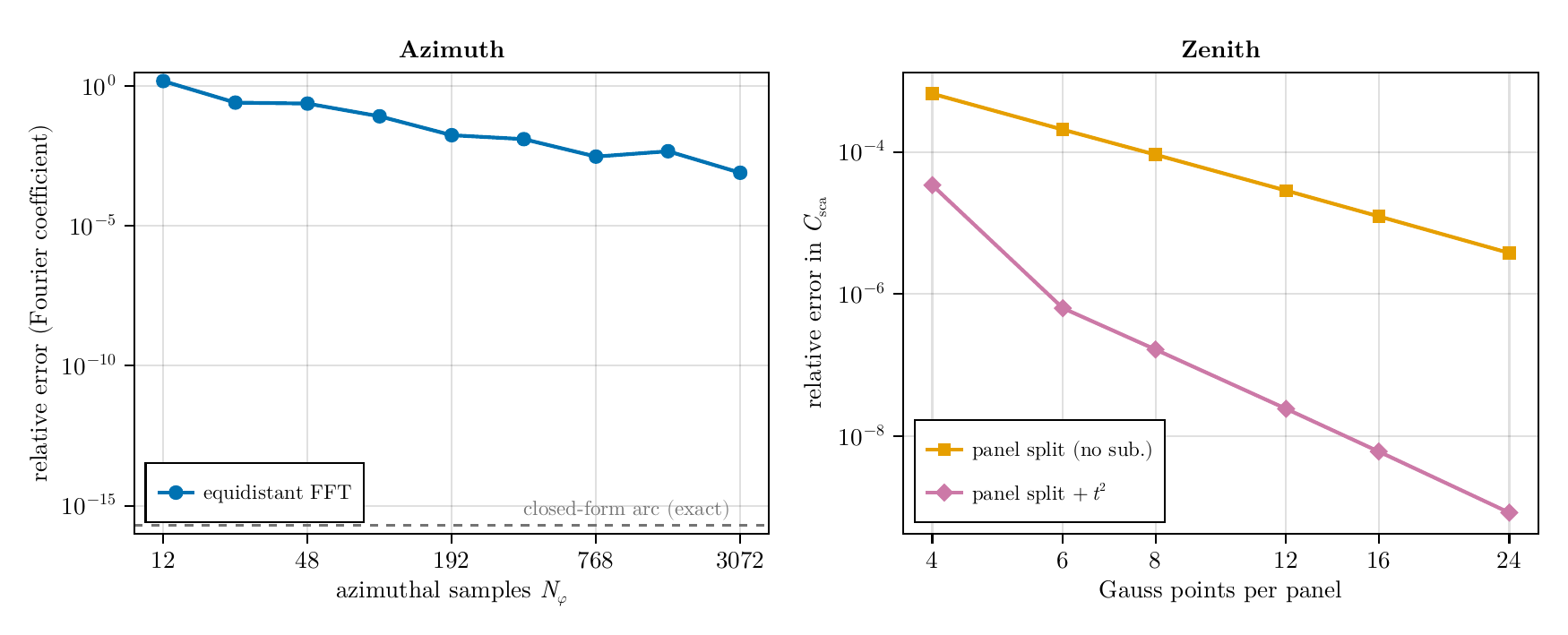}
  \caption{Angular quadrature convergence for the hexagonal prism (radial resolution
  fixed); the two panels are separate studies with independent axes. \emph{Azimuth} (abscissa $N_\varphi$; ordinate the relative error of
  the Fourier contrast coefficient): the equidistant-FFT coefficients converge as
  $\Order(1/N_\varphi)$ (the staircasing of a boxcar,
  Proposition~\ref{prop:fft}), while the closed-form arc coefficients
  (Lemma~\ref{lem:arc}) are exact (dashed floor at $\sim$\num{e-16}).
  \emph{Zenith} (abscissa Gauss points per panel; ordinate the relative error in
  $C_{\mathrm{sca}}$): panel-splitting at the inscribed-cylinder crossing alone gives
  only algebraic $\Order(N^{-3})$ (the square-root branch of Lemma~\ref{lem:sqrt},
  Proposition~\ref{prop:gauss}), whereas the substitution $\vartheta=\vartheta_b+t^2$
  (Theorem~\ref{thm:zenith}) restores spectral convergence, reaching $\sim$\num{e-9} by
  $\approx16$ nodes per panel and still falling geometrically.}
  \label{fig:angular}
\end{figure}

\begin{figure}[htbp]
  \centering
  \includegraphics[width=\textwidth]{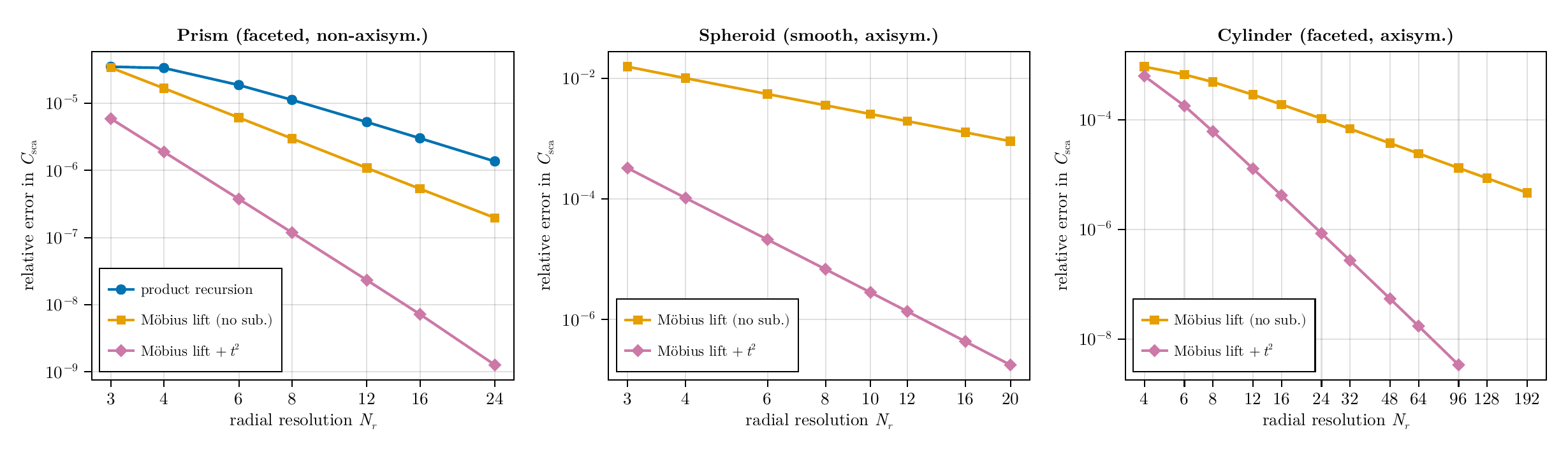}
  \caption{Radial convergence (all panels: relative error in $C_{\mathrm{sca}}$ versus
  the radial resolution $N_r$ at fixed, machine-precision angular quadrature). All
  schemes are compared at equal radial work ($N_r$ degrees of freedom); the M\"obius
  curves labeled ``no sub.'' are the linear-lift integrator \emph{without} the $t^2$
  substitution (not a different scheme). The three particles isolate the effect:
  faceted non-axisymmetric, smooth axisymmetric, and faceted axisymmetric.
  \emph{Prism:} the first-order product recursion (slope $\approx\confOrderPrismProduct$) is lifted by the
  M\"obius integrator to $\approx\confOrderPrismMobius$ (capped by the $(r-R_c)^{3/2}$ edge branch,
  Proposition~\ref{prop:radial}), and reaches fourth order once $r=r_c+t^2$ removes that
  branch. \emph{Spheroid (smooth):} the equatorial $\sqrt{r-\rmin}$ branch caps the
  unsubstituted lift at $\approx\confOrderSpheroidNosub$; the substitution restores fourth order.
  \emph{Cylinder:} the lateral-wall tangency at $r=a$ is a $\sqrt{r-a}$ branch (the
  axisymmetric counterpart of the spheroid equator), capping the unsubstituted lift at a clean
  $\confOrderCylNosubOne$; $r=a+t^2$ restores fourth order. Its caps and rim are $\sigma=1$ kinks, removed by
  panel splitting.}
  \label{fig:radial}
\end{figure}

Table~\ref{tab:orders} and Figures~\ref{fig:angular}--\ref{fig:radial} summarize the
main result: in every direction the
boundary-conformal treatment lifts the convergence from the low-order algebraic rates
of the standard scheme ($\Order(1/N)$ to $\Order(N^{-\confOrderPrismProduct})$) to exact, spectral, or
fourth-order. For the azimuthal coefficients the closed form
\eqref{eq:arc} matches a reference to $\sim$\num{e-16} while the FFT error decays as
$\Order(1/N_\varphi)$. For the zenithal integral the panel-split-plus-$t^2$ scheme
reaches $\sim$\num{e-9} by $\approx 16$ Gauss points per panel, where unsplit plain
Gauss--Legendre (not plotted) converges only at the erratic $\Order(1/N)$ staircasing,
$\sim$\num{\confZenithPlainSixForty} at $N_\vartheta=640$. For the radial
direction the matrix M\"obius integrator is first validated against the standard
product recursion (both converge to the same cross-section to $\sim$\num{\confMobiusProductAgree},
confirming the linearization), then shown to converge at order $\approx \confOrderPrismMobius$ on the
prism with panel splitting alone, the $(r-R_c)^{3/2}$ edge branch, and at order
$\approx \confOrderFourth$ once the $t^2$ substitution is added. On the smooth spheroid, where
only the equatorial $\sqrt{r-\rmin}$ branch is present, the order goes from $\confOrderSpheroidNosub$
(plain) to $\confOrderSpheroidSubTwo$ ($t^2$-substituted), confirming that the cap is the geometric
branch and not the integrator. The finite \emph{cylinder} (radius $a=1$, height
$h=1$) is the faceted axisymmetric case: its azimuthal integral is trivially exact, and the
radial integrand carries a single half-integer branch $\sqrt{r-a}$ at the lateral-wall
tangency $r=a$ (a \emph{smooth} (osculating) tangency along a circle, the axisymmetric
analogue of the spheroid equator, Proposition~\ref{prop:radial} curve case), while its
top/bottom caps ($r=h/2$) and circular rim ($r=\sqrt{a^2+(h/2)^2}$) are only $\sigma=1$ kinks,
the rim a convex \emph{edge} along a circle whose inside band pinches linearly. Wall, rim, and
the prism's isolated-point vertical edge thus realize the three excluded measures $\sqrt\xi$,
$\xi$, $\xi^{3/2}$ of \ref{app:branches} on one pair of shapes. Splitting alone gives a clean
order $\confOrderCylNosubOne$ (the local slope settles at $\confOrderCylNosubTwo$ across $N_r\in[24,192]$), and $r=a+t^2$
restores fourth order ($p\to\confOrderFourth$ by $N_r\approx96$), the conformal cylinder agreeing with the
package's independent IITM to within its discretization error.

\paragraph{Localization of the radial cap} A bisection of the radial range isolates
the order drop to the interval immediately above the lateral-edge tangency
$R_c$: the order is $\approx \confLocalOrderBelowRc$ up to $R_c$ and falls to $\approx \confLocalOrderAboveRc$ on entering
$(R_c, R_c+\confLocalWindowAboveRc)$, with no additional critical radius detected, consistent with a
$C^1$ branch (a $3/2$ power) at $R_c$ rather than a slope kink.

\paragraph{Aspect ratio} The locus of the substitution moves with the aspect ratio
(Fig.~\ref{fig:cases}), but the branch exponents and the restored order do not. The
equatorial $\sqrt{r-a}$ branch of a spheroid sits at the inner radius $\rmin=a$ (prolate,
$a<c$) or the outer radius $\rmax=a$ (oblate, $a>c$); the lateral $\sqrt{r-a}$ branch of a
cylinder is interior when flat ($a>h/2$) and coincides with $\rmin$ when tall ($a<h/2$); and
the prism's critical radii $\{b,\,R_c,\,h/2,\,\sqrt{b^2+(h/2)^2},\,\sqrt{R_c^2+(h/2)^2}\}$
reorder between a flat plate and a tall column, with $b,h/2$ face tangencies ($\sigma=1$,
split), both $R_c$ and $\sqrt{b^2+(h/2)^2}$ straight $\tfrac32$-edges, and
$\sqrt{R_c^2+(h/2)^2}$ the vertex (the higher-codimension corner of \S\ref{sec:discussion},
treated as a split point). The breakpoints are fixed analytically and enumerated within the
radial range, so the substitution lands without tuning: the restored order is the same across
orientations and aspect ratios (oblate spheroid $p\approx\confAspectOblateSpheroid$, tall cylinder $p\approx\confAspectTallCylinder$,
flat and tall prisms $p\approx\confAspectFlatPrism$ and $\confAspectTallPrism$), each validated against the package IITM, while
the unsubstituted scheme stays capped at the branch order ($\approx\confOrderCylNosubOne$ for the
$\tfrac12$-branch of the spheroid and cylinder, $\approx\confOrderPrismMobius$ for the prism's $\tfrac32$ edge).

\begin{figure}[htbp]
  \centering
  \includegraphics[width=\textwidth]{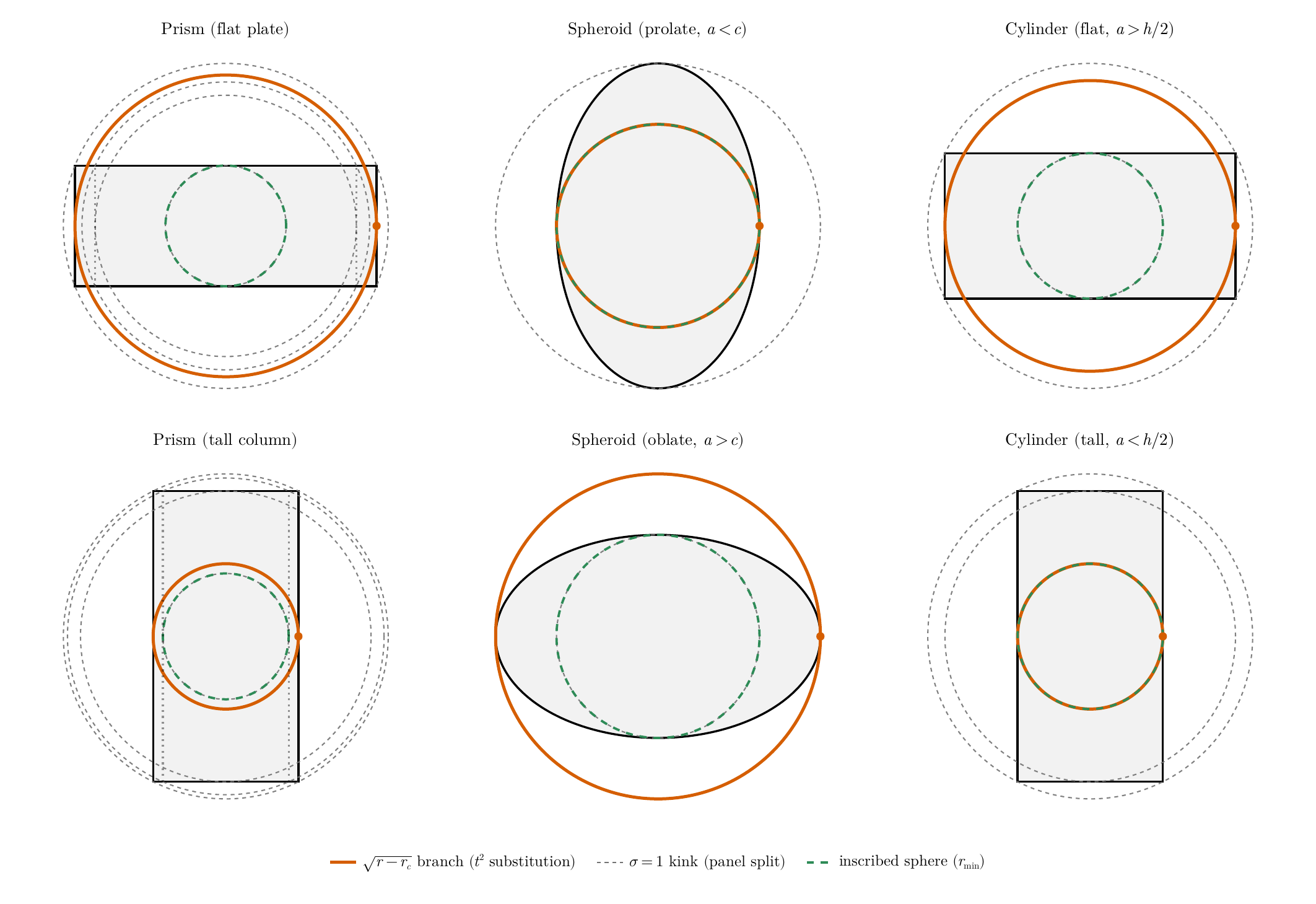}
  \caption{The substitution locus across shapes and aspect ratios. Each panel is a
  meridian cross-section with the nested integration-sphere radii at which the shell
  becomes tangent to a feature. The $\sqrt{r-r_c}$ branch absorbed by the $t^2$
  substitution (red) sits at the \emph{inner} radius $\rmin$ (prolate spheroid, tall
  cylinder), the \emph{outer} radius $\rmax$ (oblate spheroid), or an \emph{interior}
  radius (flat cylinder), and at the vertical-edge radius $R_c$ for the prism; the
  $\sigma=1$ kinks (grey, dashed) are removed by panel splitting, the inscribed sphere
  $\rmin$ is the green dashed circle, and the red dot marks the tangency point. The locus
  moves with the aspect ratio but is always fixed
  analytically by the geometry, so the scheme applies unchanged.}
  \label{fig:cases}
\end{figure}

\paragraph{Generalization to a faceted ice habit: the solid bullet}
The contact-geometry analysis of \S\ref{sec:polyhedra} holds for any convex polyhedron, so
the scheme extends from the regular prism to the pristine ice habits that motivate the IITM.
We demonstrate this on the \emph{solid hexagonal bullet}, a hexagonal column (edge
$a=1$, length $L_c=2$) capped by a hexagonal pyramid (height $h_p=1.5$, faces tilted at
$\alpha_p=\arctan(h_p/b)=60^\circ$), refractive index $m=1.311$ (ice in the near infrared,
\citep{warren2008}), $k\rmax=2.5$,
$n_{\max}=10$, which is convex but adds \emph{tilted} faces
($\alpha_f\notin\{0,\tfrac\pi2,\pi\}$), the feature the prism lacks. The solid bullet is the
pristine, convex idealization of the habit; the most common atmospheric bullets and columns have
\emph{hollow} ends \citep{schmitt2007hollow}, which are non-convex and therefore fall in the
deferred class of \S\ref{sec:discussion}. Two checks pin the
implementation: degenerating the pyramid ($h_p\to0$) reproduces the prism solver to
$\sim$\num{e-11}, and the converged cross-section agrees with the package's
\emph{independent} volume-grid IITM, which, erratic and edge-limited, brackets it to
its own $\sim$\num{e-4} discretization floor. The reference
$C_{\mathrm{sca}}^\star=\confCscaStar$ is an \emph{extrapolated} reference (Aitken/Richardson of the
conformal sequence; successive \emph{asymptotic} triples agree to $\lesssim10^{-7}$, its
self-consistency error bar), cross-validated against the independent IITM to that method's
$\sim$\num{e-4} floor.

With Proposition~\ref{prop:poly}'s breakpoints, both the \emph{column} inscribed-cylinder
branch and the \emph{tilted pyramid-face} branches take the substitution \eqref{eq:sub}
(splitting the equatorial double-branch panel at its midpoint, Remark~\ref{rem:poly}).
Figure~\ref{fig:bullet} shows the result on the tilted faces. The zenithal integral, only
$\Order(N^{-3})$ with panel splitting alone (Proposition~\ref{prop:gauss}; the unabsorbed
pyramid-face branches), becomes \emph{spectral} once substituted, reaching $\sim$\num{e-9}
by $\approx16$ Gauss points per panel (and $\sim$\num{e-10} by $\approx28$). The radial M\"obius lift recovers $p\approx\confOrderFourth$ with the
$t^2$ substitution, reaching $\sim$\num{e-8} at $N_r=32$. Without the substitution the order is an \emph{effective} $\approx\confOrderBulletNosub$, a finite-range mean drifting
toward the proven $\sigma+1=2.5$ rate; the supporting diagnostics (the substituted control's
clean integer fit, the small-coefficient half-integer fit, the cross-scheme floor, and the
comparison with the independent IITM's edge-limited floor) are collected in the Supplementary
Material. The direct evidence for the $\sigma+1=2.5$ cap is the prism panel
(Fig.~\ref{fig:radial}); on the tilted faces the bullet demonstrates that the substitution restores
fourth order, the cap there \emph{inferred} rather than directly displayed. The restored orders
match the prism's, so the generalization to tilted faces costs only the enumeration of the extra
tangency loci, not any loss of convergence order.

\begin{figure}[htbp]
  \centering
  \includegraphics[width=\textwidth]{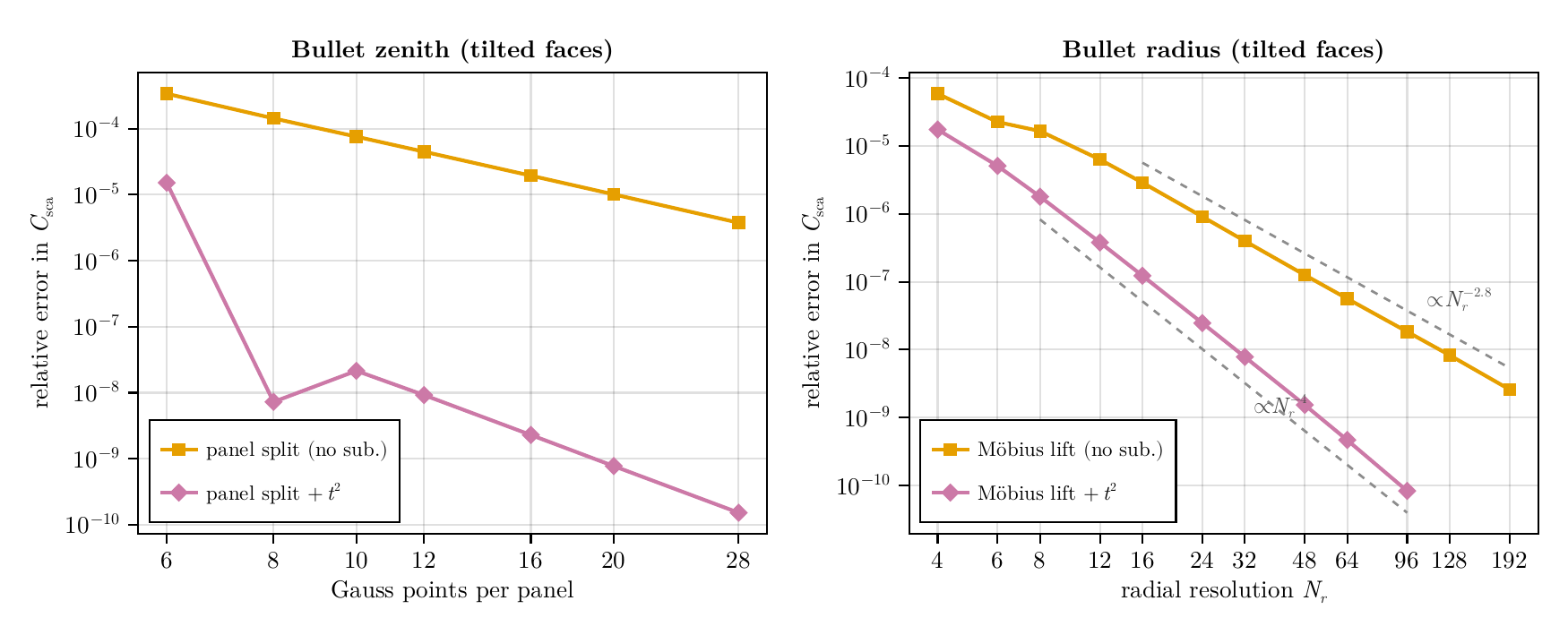}
  \caption{Boundary-conformal convergence on the solid hexagonal bullet (a convex ice habit
  with tilted pyramid faces; $a=1$, $L_c=2$, $h_p=1.5$, $m=1.311$, $k\rmax=2.5$,
  $n_{\max}=10$), relative error in $C_{\mathrm{sca}}$ against the extrapolated reference
  $C_{\mathrm{sca}}^\star$. \emph{Zenith} (Gauss points per panel): panel splitting alone leaves the tilted
  pyramid-face square-root branches unabsorbed and converges only algebraically
  ($\Order(N^{-3})$, Proposition~\ref{prop:gauss}); the substitution \eqref{eq:sub}, applied
  at the column \emph{and} pyramid tangencies (Proposition~\ref{prop:poly}, Remark~\ref{rem:poly}),
  restores spectral convergence. \emph{Radius} ($N_r$): the M\"obius lift reaches $p\approx\confOrderFourth$
  once $r=r_c+t^2$ absorbs the half-integer edge branches (lower dashed guide); without the
  substitution the order is an \emph{effective} $\approx\confOrderBulletNosub$ (upper dashed guide; the
  finite-range fit and diagnostics are in the Supplementary Material). Low-resolution points are pre-asymptotic
  (sign-oscillation in the zenith, accidental cancellation in the radius) and excluded from the
  order fits.}
  \label{fig:bullet}
\end{figure}

\section{Discussion}
\label{sec:discussion}

The practical message for IITM users is that the angular and radial grids do not
have to be refined to fight staircasing: the loci of non-smoothness are known
analytically (they are the tangency angles and radii), and once the integration is
made conformal to them the convergence is governed by the smooth parts of the
integrand. The unifying mechanism (jumps, kinks, and half-integer branches set by the tangency
geometry, removed respectively by exact arc integration, panel splitting, and the
square-root substitution) ties together the azimuthal, zenithal and radial
treatments that would otherwise look unrelated.

For the angular direction, the refinement over the interval-splitting of \citet{zhai2019} is the
branch-type identification: the crossing is a square-root branch (Lemma~\ref{lem:sqrt}), splitting
alone is $\Order(N^{-3})$ (Proposition~\ref{prop:gauss}), and the substitution restores spectral
convergence (Theorem~\ref{thm:zenith}). At the coarser accuracies of routine phase-function work the
angular cost is dominated by the oscillation of high-order Wigner-$d$ functions
($N_q\approx600$--$1000$, Remark~\ref{rem:unif}) and panel splitting alone suffices; the contrast
square-root branch becomes the binding constraint only at higher accuracy, where the substitution is
therefore the enabling step rather than an incremental speedup. A production-tuned, fully
co-converged end-to-end cost study is beyond the present scope: at fixed $n_{\max}$ a timing
comparison is physically informative only down to the multipole-truncation error; whatever
speedup accrues below it buys digits of the fixed-$n_{\max}$ limit, not of the physical
cross-section. A meaningful study must therefore co-converge the truncation
with the quadrature, and is left to the framework paper \citep{transitionmatrices}; the
present results establish the convergence \emph{orders}. To close the question of whether the
scheme carries any cost penalty, however, Table~\ref{tab:timing} gives an indicative
equal-accuracy wall-clock comparison on the solid bullet. At the plain scheme's best attainable
accuracy ($\sim$\num{2e-4}, its edge-limited floor) the boundary-conformal scheme already
reaches the same cross-section faster ($N_r=3$: \SI{0.22}{s} versus \SI{0.27}{s}) and two
orders more accurately ($\sim$\num{e-6}), continuing to accuracies the plain scheme cannot
reach at this truncation. The
comparison is indicative only: it is a single-thread, fixed-$n_{\max}$ measurement on one machine,
not the co-converged study above.

\begin{table}[htbp]
\centering
\caption{Indicative equal-accuracy wall-clock comparison on the solid hexagonal bullet
($k\rmax=2.5$, $n_{\max}=10$), relative error in $C_{\mathrm{sca}}$ against the extrapolated
reference $C_{\mathrm{sca}}^\star$. Single thread (one BLAS thread, single-process Julia),
best-of-many-samples timing on an Apple-silicon laptop; absolute times are machine-dependent and
the comparison is \emph{indicative}, not a production cost study. The plain scheme is
azimuth-resolution-insensitive here (edge-, not azimuth-limited), so its fastest converged azimuth
is used. At its best attainable accuracy the boundary-conformal scheme matches the plain scheme in
comparable time, then continues to accuracies the plain scheme does not reach at this truncation.
The wall-clock entries below are indicative single-thread timings on this machine and are not
regenerated by the figure pipeline.}
\label{tab:timing}
{\small
\begin{tabular}{lll}
\toprule
Scheme (resolution) & Relative error & Wall-clock (s) \\
\midrule
Plain IITM ($N_r,N_\vartheta,N_\varphi=80,60,96$) & $2.2\times10^{-4}$ & $0.27$ \\
Boundary-conformal ($N_r=3$)                       & $2.7\times10^{-6}$ & $0.22$ \\
Boundary-conformal ($N_r=8$)                       & $1.8\times10^{-6}$ & $0.86$ \\
\bottomrule
\end{tabular}}
\end{table}

On the radial side the high-order recurrence is \citet{doicu2019}'s; ours is the observation
that faceted particles carry half-integer branches that cap it at order $\sigma+1$ with no
visible warning, and the substitution that restores the design order (\S\ref{sec:radord}). The remaining
limitations are twofold. The first is the higher-codimension coincidences where two tangency
loci merge (the sphere passing through a particle corner); these occur at isolated radii and
are not analyzed here; in the tested geometries we treat them as split points, and they do
not control the observed order. The second is the extension to \emph{non-convex} habits
(hollow columns, bullet rosettes, and aggregates), whose re-entrant boundaries break the
half-space inside-test of \S\ref{sec:polyhedra} and would require ray-cast parity or a convex
decomposition. The
convex case, the pristine prismatic habits (column, plate, droxtal, solid bullet), is
covered: each face contributes its own analytic arcs and tangency branches through
Proposition~\ref{prop:poly}, demonstrated here on the bullet.

\paragraph{Refractive index and absorption} The contrast enters every branch coefficient
only through the scalar prefactor $m^2-1$, whereas the branch exponents
($\sigma=\tfrac32,1,\tfrac12$) and the restored orders are fixed by the contact geometry
alone (Proposition~\ref{prop:radial}). A complex (absorbing) index makes $m^2-1$ complex
without altering either the branch structure or the restored orders, so the classification
and the $t^2$ substitution carry over verbatim to absorbing habits. The faceted examples
here use real indices (the ice value $m=1.311$ is essentially lossless in the near
infrared); this is the conservative choice, since absorption damps the interior field and
only eases the conditioning of the contrast quadrature.

\paragraph{Generality} Coordinate transformations that cancel singularities are a long-standing
device in computational integral equations (the self-adaptive transformation of
\citet{telles1987} and the vertex-cancelling map of \citet{duffy1982} are the canonical
boundary-element examples), and the $t^2$ substitution used here is the simplest such map. What
is specific to the present work is not the substitution but the \emph{a priori} contact-geometry
classification that tells one, in closed form, which branch sits where. That classification is a
property of the \emph{contrast quadrature over a faceted scatterer}, not of the imbedding
recursion: the same sharp-edge tangency arises in the surface integrals of the extended
boundary condition method, where \citet{kahnert2001} likewise had to treat the edge
discontinuity by adapting the surface integration. We therefore expect the
contact-geometry classification (Proposition~\ref{prop:radial}) and the square-root
substitution to carry over to EBCM-type surface quadratures for sharp-edged particles,
though the present proofs are specific to the IITM moving-shell contrast integral, and EBCM
surface integrals additionally carry the physical Meixner edge singularity (a different
analytic object from the bounded, piecewise-constant contrast), so the carry-over requires
a separate analysis. The radial component, in
turn, is an instance of structure-preserving integration of an invariant-imbedding
matrix Riccati equation; its scalar counterpart is the variable-phase
(phase-function) method of potential scattering \citep{calogero1967}, and the same
linear lift carries to transfer- and scattering-matrix layer recurrences. The IITM
results are thus the worked example of a broader principle: half-integer
geometric-contact branches in moving-boundary quadratures and in layer recursions,
removed by a power substitution to recover high-order convergence.

\section{Conclusion}
\label{sec:conclusion}

We have shown that IITM staircasing for faceted particles is a manifestation, in all
three integration directions, of jumps, kinks, and half-integer branches induced by the
tangency geometry of the integration sphere against the faces and edges of the particle.
A boundary-conformal integration scheme (closed-form azimuthal arcs, panel
splitting at the analytic tangency loci, and a square-root substitution that removes the
half-integer branches) restores exact, spectral, or fourth-order convergence on a
hexagonal prism, a finite cylinder, a prolate spheroid, and, through a convex-polyhedron
enumeration of the tangency loci, the solid hexagonal bullet, a faceted ice habit with
tilted faces. The scheme turns the geometry of the particle from the source of the error
into the data that removes it.

\section*{Data availability}
The boundary-conformal scheme builds on the open-source
\texttt{Transition\allowbreak Matrices.jl} package \citep{transitionmatrices}, whose baseline
IITM and EBCM solvers it extends. A reference
implementation of the conformal solvers (shared library \texttt{src/}) and the scripts
that regenerate every figure (directory \texttt{conformal/}) are available at
\url{https://github.com/JuliaRemoteSensing/iitm-convergence}
(release \texttt{conformal-v1.0}); the repository pins the computational environment
(\texttt{Project.toml}/\texttt{Manifest.toml}, Julia~1.12) at that tag. The conformal solvers there
are self-contained; for direct comparison with \eqref{eq:arc}, the implementation uses
the conjugate azimuthal convention ($e^{-\mathrm iq\varphi}$, same $q=m'-m$), so its Fourier
coefficients are the complex conjugates of \eqref{eq:arc} and the cross-sections are identical.

\appendix

\section{Detailed local branch derivations}
\label{app:branches}
This appendix gives the leading-order local contact model behind
Proposition~\ref{prop:radial}: the exponent $\sigma$ and coefficient $C$ are exact, the remainder
is controlled by the rescaling below (Remark~\ref{rem:status}), and the degenerate cases are
excluded by hypothesis.

\paragraph{Edge case (model computation)} At a vertical edge ($N$-gon vertex) the polar
boundary has a linear corner, $R(\varphi)=R_c(1-\tau|\varphi|)$. With $\rho=r\cos\eta$,
$\cos\eta=1-\tfrac12 \eta^2+\Order(\eta^4)$, and $r=R_c(1+\xi)$, the inside condition
$\rho\le R(\varphi)$ reduces to $\tfrac12 \eta^2-\xi\ge\tau|\varphi|$, and the
region that leaves the particle for small $\xi>0$ is
$\{\tfrac12 \eta^2<\xi+\tau|\varphi|\}$, of measure
\[
  A(\xi)=\!\int_{-\varphi_0}^{\varphi_0}\!\!2\sqrt{2(\xi+\tau|\varphi|)}\,d\varphi
   =\frac{8\sqrt2}{3\tau}\Bigl[(\xi+\tau\varphi_0)^{3/2}-\xi^{3/2}\Bigr].
\]
Here $(\xi+\tau\varphi_0)^{3/2}$ is analytic in $\xi$ (binomial series,
$\tau\varphi_0$ bounded away from $0$) and the cutoff-independent non-analytic part is
exactly $-\tfrac{8\sqrt2}{3\tau}\xi^{3/2}$.

\paragraph{Corrections (rescaling)} The non-analytic part is generated in the inner region
$\eta=\Order(\sqrt\xi)$, $\varphi=\Order(\xi)$. Rescaling
$\eta=\sqrt{2\xi}\,\alpha$, $\varphi=(\xi/\tau)\,\beta$, the exact boundary
$r\cos\eta=R(\varphi)$ becomes $\alpha^2=1+|\beta|+\Order(\xi)$, the measure
$\sin\vartheta\,d\eta\,d\varphi=(\sqrt2/\tau)\,\xi^{3/2}\,(1+\Order(\xi))\,
d\alpha\,d\beta$, and the integrand is $w_0+\Order(\sqrt\xi)$. Thus the neglected
boundary curvature $\Order(\varphi^2)$, the $\cos\eta$ and $\sin\vartheta$ tails, and the
integrand variation each carry an extra positive power of $\xi$ and contribute
at $o(\xi^{3/2})$, while the outer region $\varphi=\Order(1)$ is analytic in
$\xi$. The excluded measure $A(\xi)$ thus carries the non-analytic part
$-\tfrac{8\sqrt2}{3\tau}\xi^{3/2}$; the retained integral $I=\int_{\mathrm{inside}}w$
loses it, so $I$ acquires $+\tfrac{8\sqrt2}{3\tau}w_0\,\xi^{3/2}$. This gives
\eqref{eq:branchexp} with $\sigma=\tfrac32$ and $|C|=\tfrac{8\sqrt2}{3\tau}|w_0|$.

\paragraph{Face case} A lateral-face tangency replaces the linear corner by a smooth
quadratic minimum of the polar boundary, $R(\varphi)=b(1+c\varphi^2)$, $c>0$. The
inside condition becomes $\tfrac12 \eta^2+c\varphi^2\ge\xi$, so the excluded set
$\{\tfrac12 \eta^2+c\varphi^2<\xi\}$ is an ellipse of measure
$\pi\sqrt{2/c}\,\xi$ for $\xi>0$ and $0$ for $\xi<0$: each one-sided
branch is analytic, but the two join with a slope discontinuity at $\xi=0$, hence
$\sigma=1$ (a slope kink, removed by splitting at $r_c$).

\paragraph{Curve-tangency case} For a \emph{smooth} (osculating) tangency along a curve (the
prolate spheroid's equator, with $r_{\mathrm{surf}}(\vartheta)=\rmin(1+c\,\eta^2)$
\emph{uniformly in $\varphi$}), the inside condition $r\le r_{\mathrm{surf}}$ reads
$c\,\eta^2\ge\xi$,
so the excluded set is the equatorial band $\{|\eta|<\sqrt{\xi/c}\}$ over the full
azimuth, of measure $2\pi\cdot2\sqrt{\xi/c}=\Order(\xi^{1/2})$; hence
$\sigma=\tfrac12$. This is the $\sqrt{r-\rmin}$ branch removed in the spheroid example by
$r=\rmin+t^2$. (As for the edge, the leading coefficient is nonzero for generic
observables.)

\section{Convex-polyhedron breakpoint enumeration}
\label{app:poly}
\emph{Proof of Proposition~\ref{prop:poly}.}
Substituting $\hat u$ into $\mathbf n_f\!\cdot x\le d_f$ gives \eqref{eq:facecond} directly.
(i) For a non-axial face ($A_f\neq0$), $|B_f/A_f|<1$ makes $\cos(\varphi-\beta_f)=B_f/A_f$ have
the two roots $\varphi=\beta_f\pm\psi_f$, between which the face excludes the circle; the
surviving set is a union of arcs and \eqref{eq:arc} integrates the piecewise-constant contrast
over it exactly. An axial face ($A_f=0$) makes \eqref{eq:facecond} read
$\cos\alpha_f\cos\vartheta\le d_f/r$, independent of $\varphi$, so it gates the entire circle. (ii) The $\varphi$-circle is tangent to a face where its extreme value
$\max_\varphi[\sin\alpha_f\sin\vartheta\cos(\varphi-\beta_f)]+\cos\alpha_f\cos\vartheta
=\cos(\vartheta-\alpha_f)$ (or the min, $\cos(\vartheta+\alpha_f)$) equals $d_f/r$, i.e.\ at
$\vartheta=\pm\alpha_f\pm\acos(d_f/r)$; an edge contributes where its line pierces the
sphere. A tangency to a face \emph{plane} on which that face is not part of $\partial\Omega$
is spurious and removed by the on-boundary test. The local branch type follows from the
order of contact as in Lemma~\ref{lem:zbreak} and Proposition~\ref{prop:radial}: a face
is a smooth quadratic extremum of the polar boundary ($\sigma=1$), a straight edge a linear
corner ($\sigma=\tfrac32$). (iii) The structure of (ii) changes only when the sphere starts
or stops touching a face ($r=d_f$), passes a vertex ($r=|v|$), or becomes tangent to an
edge line.\hfill$\square$

\begin{remark}[Reduction and a double-branch subtlety]\label{rem:poly}
For axial caps ($\alpha_f\in\{0,\pi\}$) and vertical sides ($\alpha_f=\tfrac\pi2$),
\eqref{eq:facecond} and Proposition~\ref{prop:poly} reduce to the regular-prism loci
$\vartheta_{\mathrm{cap}}=\acos(h/(2r))$, $\vartheta_b=\arcsin(b/r)$,
$\vartheta_{R_c}=\arcsin(R_c/r)$ and the critical radii of \S\ref{sec:radial}. We verified
the general loci against brute-force point-in-polyhedron tests for a solid hexagonal
bullet (the formula check is geometry-independent): the analytic $\varphi$-arc endpoints match to
$\confArcEndpointMatch$ and the filtered zenithal breakpoint set matches the measured
non-smoothness loci exactly. One subtlety the single-face view misses: a partial-arc zenith
panel can be bounded by a square-root tangency at \emph{both} ends (the equatorial belt
when $r<R_c$, where no circumradius crossing separates the two inscribed tangencies), and
a single substitution \eqref{eq:sub} cannot absorb both; splitting such a panel at its
midpoint leaves one branch per panel and restores spectral convergence (\S\ref{sec:results}).
Each resulting square-root panel is then mapped by $x=x_c+t^2$ when its branch sits at the lower
endpoint and by $x=x_c-t^2$ when it sits at the upper endpoint, $x_c$ being that endpoint.
\end{remark}

\section{Radial integrator order and the RK4 lift}
\label{app:radial}
\emph{Proof of Lemma~\ref{lem:orderred}.}
The singularity is in the coefficient, not the solution: integrating once,
$\mathbf y(r)=\mathbf y_{\mathrm{an}}(r)+(r-r_c)^{\sigma+1}\mathbf w(r)$ with
$\mathbf y_{\mathrm{an}}$ analytic and $(r-r_c)^{\sigma+1}\mathbf w(r)$ the leading
non-analytic term ($\mathbf w$ bounded; the nonlinear feedback adds only higher-order
branches $(r-r_c)^{2\sigma+1},\dots$). What the estimate uses is the resulting
derivative growth $\mathbf y^{(k)}(r)=\Order((r-r_c)^{\sigma+1-k})$, which the leading
term already fixes.
A one-step method of classical order $p$ applied to a solution of this finite smoothness
attains global order $\min(p,\sigma+1)$, not $p$ (the standard order reduction for
one-step methods on integrands/solutions of limited regularity \citep{hairer_ode1}).
Quantitatively, the local error on the step $[r_c+jh,\,r_c+(j+1)h]$ is
$\Order(h^{p+1}\|\mathbf y^{(p+1)}\|_\infty)=\Order\!\big(h^{p+1}(jh)^{\sigma-p}\big)$ for
$j\ge1$ and $\Order(h^{\sigma+1})$ on the endpoint step; summing,
$\sum_{j\ge1}h^{p+1}(jh)^{\sigma-p}+\Order(h^{\sigma+1})=\Order\!\big(h^{\min(p,\sigma+1)}\big)$
(for $\sigma<p-1$ the series converges and the $\Order(h^{\sigma+1})$ endpoint term
dominates; otherwise the $\Order(h^{p})$ bulk does, and the borderline $\sigma=p-1$ does not
arise, $\sigma$ being non-integer). Finally the map
$(\mathbf P,\mathbf V)\mapsto\mathbf V\mathbf P^{-1}$ is analytic, hence Lipschitz, on the
compact trajectory where $\mathbf P$ is invertible, so it transfers the order to $\Tmat$.\hfill$\square$

\paragraph{Subleading logarithm} The substitution maps the leading half-integer $\sigma$ to an
integer power and restores fourth order for the RK4 lift; we do not claim full analyticity. A
subleading $\xi^2\log\xi$ from inner/outer matching, if present
(Remark~\ref{rem:status}), maps to $t^4\log t$. This does not change the leading algebraic order; it
introduces at most a logarithmic factor that is indistinguishable from clean fourth order over the
fitted resolution range, and fourth order is what is observed directly in \S\ref{sec:results}.

\section*{Declaration of generative AI and AI-assisted technologies in the
          writing process}
During the preparation of this work the authors used Claude (Anthropic) to draft and
edit text. After using this tool, the authors reviewed and edited the content as needed
and take full responsibility for the content of the publication.

\bibliographystyle{elsarticle-num-names}
\bibliography{references}

@article{waterman1971,
  author = {Waterman, P. C.},
  title = {Symmetry, Unitarity, and Geometry in Electromagnetic Scattering},
  journal = {Physical Review D},
  volume = {3}, number = {4}, pages = {825--839}, year = {1971},
  doi = {10.1103/PhysRevD.3.825}
}

@book{mishchenko2002,
  author = {Mishchenko, Michael I. and Travis, Larry D. and Lacis, Andrew A.},
  title = {Scattering, Absorption, and Emission of Light by Small Particles},
  publisher = {Cambridge University Press}, address = {Cambridge}, year = {2002}
}

@article{johnson1988,
  author = {Johnson, B. R.},
  title = {Invariant imbedding {T} matrix approach to electromagnetic scattering},
  journal = {Applied Optics},
  volume = {27}, number = {23}, pages = {4861--4873}, year = {1988},
  doi = {10.1364/AO.27.004861}
}

@article{bi2013,
  author = {Bi, Lei and Yang, Ping and Kattawar, George W. and Mishchenko, Michael I.},
  title = {Efficient implementation of the invariant imbedding {T}-matrix method
           and the separation of variables method applied to large nonspherical
           inhomogeneous particles},
  journal = {Journal of Quantitative Spectroscopy and Radiative Transfer},
  volume = {116}, pages = {169--183}, year = {2013},
  doi = {10.1016/j.jqsrt.2012.11.014}
}

@article{bi2014,
  author = {Bi, Lei and Yang, Ping},
  title = {Accurate simulation of the optical properties of atmospheric ice
           crystals with the invariant imbedding {T}-matrix method},
  journal = {Journal of Quantitative Spectroscopy and Radiative Transfer},
  volume = {138}, pages = {17--35}, year = {2014},
  doi = {10.1016/j.jqsrt.2014.01.013}
}

@book{sun2020iitm,
  author = {Sun, Bingqiang and Bi, Lei and Yang, Ping and Kahnert, Michael and
            Kattawar, George W.},
  title = {Invariant Imbedding {T}-matrix Method for Light Scattering by
           Nonspherical and Inhomogeneous Particles},
  publisher = {Elsevier}, address = {Amsterdam}, year = {2020},
  doi = {10.1016/C2018-0-02999-0}
}

@article{doicu2019,
  author = {Doicu, Adrian and Wriedt, Thomas and Khebbache, Naima},
  title = {An overview of the methods for deriving recurrence relations for
           {T}-matrix calculation},
  journal = {Journal of Quantitative Spectroscopy and Radiative Transfer},
  volume = {224}, pages = {289--302}, year = {2019},
  doi = {10.1016/j.jqsrt.2018.11.029}
}

@misc{transitionmatrices,
  author = {Xiong, Yu and Wu, Zihua},
  title = {{TransitionMatrices.jl}: a high-efficiency open-source {Julia} framework
           for electromagnetic scattering and polarimetric remote sensing of
           nonspherical atmospheric particles},
  year = {2025},
  note = {SSRN preprint, \url{https://doi.org/10.2139/ssrn.5335549}; code:
          \url{https://github.com/JuliaRemoteSensing/TransitionMatrices.jl}}
}

@article{hu2020,
  author = {Hu, Shuai and Liu, Lei and Gao, Taichang and Zeng, Qingwei},
  title = {An analysis of the factors influencing the modeling accuracy of the
           invariant imbedding {T}-matrix method and the optimal design of the
           parameter settings for particles with different geometrical and optical
           properties},
  journal = {Journal of Quantitative Spectroscopy and Radiative Transfer},
  volume = {256}, pages = {107306}, year = {2020},
  doi = {10.1016/j.jqsrt.2020.107306}
}

@article{hu2021sym,
  author = {Hu, Shuai and Liu, Lei and Zeng, Qingwei and Gao, Taichang and Zhang, Feng},
  title = {An investigation of the symmetrical properties in the invariant imbedding
           {T}-matrix method for the nonspherical particles with symmetrical geometry},
  journal = {Journal of Quantitative Spectroscopy and Radiative Transfer},
  volume = {259}, pages = {107401}, year = {2021},
  doi = {10.1016/j.jqsrt.2020.107401}
}

@article{sun2021jac,
  author = {Sun, Bingqiang and Gao, Chenxu and Bi, Lei and Spurr, Robert},
  title = {Analytical {J}acobians of single scattering optical properties using the
           invariant imbedding {T}-matrix method},
  journal = {Optics Express},
  volume = {29}, number = {6}, pages = {9635--9669}, year = {2021},
  doi = {10.1364/OE.421886}
}

@article{hu2023dviim,
  author = {Hu, Shuai and Li, Shulei and Zeng, Qingwei and Liu, Lei},
  title = {Dimension-variable invariant imbedding ({DVIIM}) {T}-matrix computational
           method for the light scattering simulation of atmospheric nonspherical
           particles},
  journal = {Optics Express},
  volume = {31}, number = {6}, pages = {10052--10069}, year = {2023},
  doi = {10.1364/OE.472809}
}

@article{mishchenko2020db,
  author = {Mishchenko, Michael I.},
  title = {Comprehensive thematic {T}-matrix reference database: a 2017--2019 update},
  journal = {Journal of Quantitative Spectroscopy and Radiative Transfer},
  volume = {242}, pages = {106692}, year = {2020},
  doi = {10.1016/j.jqsrt.2019.106692}
}

@article{kahnert2001,
  author = {Kahnert, F. Michael and Stamnes, Jakob J. and Stamnes, Knut},
  title = {Application of the extended boundary condition method to particles with
           sharp edges: a comparison of two surface integration approaches},
  journal = {Applied Optics},
  volume = {40}, number = {18}, pages = {3101--3109}, year = {2001},
  doi = {10.1364/AO.40.003101}
}

@article{ganeshhawkins2025,
  author = {Ganesh, M. and Hawkins, S. C.},
  title = {{T}-matrix computations for light scattering by penetrable particles with
           large aspect ratios},
  journal = {Journal of Quantitative Spectroscopy and Radiative Transfer},
  volume = {334}, pages = {109346}, year = {2025},
  doi = {10.1016/j.jqsrt.2025.109346}
}

@article{zhai2019,
  author = {Zhai, Siyao and Panetta, R. Lee and Yang, Ping},
  title = {Improvements in the computational efficiency and convergence of the
           invariant imbedding {T}-matrix method for spheroids and hexagonal prisms},
  journal = {Optics Express},
  volume = {27}, number = {20}, pages = {A1441--A1457}, year = {2019},
  doi = {10.1364/OE.27.0A1441}
}

@article{schiff1999,
  author = {Schiff, Jeremy and Shnider, S.},
  title = {A Natural Approach to the Numerical Integration of {R}iccati
           Differential Equations},
  journal = {SIAM Journal on Numerical Analysis},
  volume = {36}, number = {5}, pages = {1392--1413}, year = {1999},
  doi = {10.1137/S0036142996307946}
}

@article{dieci1996,
  author = {Dieci, Luca and Eirola, Timo},
  title = {Preserving monotonicity in the numerical solution of {R}iccati
           differential equations},
  journal = {Numerische Mathematik},
  volume = {74}, number = {1}, pages = {35--47}, year = {1996},
  doi = {10.1007/s002110050206}
}

@book{calogero1967,
  author = {Calogero, Francesco},
  title = {Variable Phase Approach to Potential Scattering},
  series = {Mathematics in Science and Engineering}, volume = {35},
  publisher = {Academic Press}, address = {New York}, year = {1967}
}

@book{hairer_ode1,
  author = {Hairer, Ernst and N{\o}rsett, Syvert P. and Wanner, Gerhard},
  title = {Solving Ordinary Differential Equations I: Nonstiff Problems},
  edition = {2nd}, series = {Springer Series in Computational Mathematics}, volume = {8},
  publisher = {Springer}, address = {Berlin}, year = {1993},
  doi = {10.1007/978-3-540-78862-1}
}

@book{trefethen2013,
  author = {Trefethen, Lloyd N.},
  title = {Approximation Theory and Approximation Practice},
  publisher = {SIAM}, address = {Philadelphia}, year = {2013}
}

@book{davis1984,
  author = {Davis, Philip J. and Rabinowitz, Philip},
  title = {Methods of Numerical Integration},
  edition = {2nd}, publisher = {Academic Press}, address = {Orlando}, year = {1984}
}

@book{hairer2006,
  author = {Hairer, Ernst and Lubich, Christian and Wanner, Gerhard},
  title = {Geometric Numerical Integration: Structure-Preserving Algorithms for
           Ordinary Differential Equations},
  edition = {2nd}, series = {Springer Series in Computational Mathematics}, volume = {31},
  publisher = {Springer}, address = {Berlin}, year = {2006},
  doi = {10.1007/3-540-30666-8}
}

@article{zhao2022sym,
  author = {Zhao, Jiaqi and Hu, Shuai and Liu, Xichuan and Li, Shulei},
  title = {The Computational Optimization of the Invariant Imbedding {T} Matrix
           Method for the Particles with {N}-Fold Symmetry},
  journal = {Remote Sensing},
  volume = {14}, number = {16}, pages = {4061}, year = {2022},
  doi = {10.3390/rs14164061}
}

@article{zhang2022vswf,
  author = {Zhang, Yuheng and Ding, Jiachen and Yang, Ping and Panetta, R. Lee},
  title = {Vector spherical wave function truncation in the invariant imbedding
           {T}-matrix method},
  journal = {Optics Express},
  volume = {30}, number = {17}, pages = {30020--30037}, year = {2022},
  doi = {10.1364/OE.459648}
}

@article{sun2022liitm,
  author = {Sun, Bingqiang and Gao, Chenxu and Liang, Dongbin and Liu, Zhaoyuan and
            Liu, Jian},
  title = {Capability and convergence of linearized invariant-imbedding {T}-matrix
           and physical-geometric optics methods for light scattering},
  journal = {Optics Express},
  volume = {30}, number = {21}, pages = {37769--37785}, year = {2022},
  doi = {10.1364/OE.473075}
}

@article{wang2023grid,
  author = {Wang, Zheng and Bi, Lei and Kong, Senyi},
  title = {Flexible implementation of the particle shape and internal inhomogeneity in
           the invariant imbedding {T}-matrix method},
  journal = {Optics Express},
  volume = {31}, number = {18}, pages = {29427--29444}, year = {2023},
  doi = {10.1364/OE.498190}
}

@article{telles1987,
  author = {Telles, J. C. F.},
  title = {A self-adaptive co-ordinate transformation for efficient numerical evaluation
           of general boundary element integrals},
  journal = {International Journal for Numerical Methods in Engineering},
  volume = {24}, number = {5}, pages = {959--973}, year = {1987},
  doi = {10.1002/nme.1620240509}
}

@article{duffy1982,
  author = {Duffy, Michael G.},
  title = {Quadrature over a pyramid or cube of integrands with a singularity at a vertex},
  journal = {SIAM Journal on Numerical Analysis},
  volume = {19}, number = {6}, pages = {1260--1262}, year = {1982},
  doi = {10.1137/0719090}
}

@article{yang2013ice,
  author = {Yang, Ping and Bi, Lei and Baum, Bryan A. and Liou, Kuo-Nan and
            Kattawar, George W. and Mishchenko, Michael I. and Cole, Benjamin},
  title = {Spectrally consistent scattering, absorption, and polarization properties of
           atmospheric ice crystals at wavelengths from 0.2 to 100 {\textmu}m},
  journal = {Journal of the Atmospheric Sciences},
  volume = {70}, number = {1}, pages = {330--347}, year = {2013},
  doi = {10.1175/JAS-D-12-039.1}
}

@article{schmitt2007hollow,
  author = {Schmitt, C. G. and Heymsfield, A. J.},
  title = {On the occurrence of hollow bullet rosette- and column-shaped ice crystals
           in midlatitude cirrus},
  journal = {Journal of the Atmospheric Sciences},
  volume = {64}, number = {12}, pages = {4515--4520}, year = {2007},
  doi = {10.1175/2007JAS2317.1}
}

@article{warren2008,
  author = {Warren, Stephen G. and Brandt, Richard E.},
  title = {Optical constants of ice from the ultraviolet to the microwave:
           A revised compilation},
  journal = {Journal of Geophysical Research: Atmospheres},
  volume = {113}, number = {D14}, pages = {D14220}, year = {2008},
  doi = {10.1029/2007JD009744}
}

\end{document}